\documentclass[journal]{IEEEtran}
\ifCLASSINFOpdf%
  \usepackage[pdftex]{graphicx}
  \DeclareGraphicsExtensions{.pdf,.jpeg,.png}
\else
\fi

\usepackage[numbers,sort&compress]{natbib}
\usepackage[inline]{enumitem}
\usepackage{amsmath}
\usepackage{amsthm}
\theoremstyle{plain}

\newtheorem{proposition}{Proposition}

\usepackage{amsmath,amssymb,amsthm}
\usepackage[colorlinks]{hyperref}
\usepackage{booktabs}
\usepackage{color}
\usepackage[dvipsnames]{xcolor}
\usepackage{cleveref}

\newcommand{\rev}[1]{{\color{black}#1}}

\DeclareMathOperator*{\e}{\mathrm{e}}

\newcommand\given[1][]{\,#1\vert\,}
\newcommand\mgiven[1][]{\,#1\middle\vert\,}

\usepackage{tikz}
\newcommand\copyrighttext{%
  \footnotesize Copyright \copyright~2018 IEEE. Personal use of this material is permitted. Permission from IEEE must be obtained for all other uses.
}
\newcommand\copyrightnotice{%
\begin{tikzpicture}[remember picture,overlay]
\node[anchor=north,yshift=0pt] at (current page.north) {\fbox{\parbox{\dimexpr\textwidth-\fboxsep-\fboxrule\relax}{\copyrighttext}}};
\end{tikzpicture}%
}

\begin{document}
%
\title{Adaptive pooling operators for weakly labeled sound event detection}
%
%
%

\author{
      Brian McFee$^{1, 3}$,
      Justin Salamon$^{1, 2}$,
      Juan Pablo Bello$^{1}$  \IEEEmembership{Senior Member,~IEEE}
\thanks{$^1$ Music and Audio Research Laboratory}%
\thanks{$^2$ Center for Urban Science and Progress}%
\thanks{$^3$ Center for Data Science}%
\thanks{New York University, New York, USA}%
}

\markboth{IEEE Transactions on Audio, Speech and Language Processing, In press, 2018}{}
\maketitle

\copyrightnotice

\begin{abstract}
Sound event detection (SED) methods are tasked with labeling segments of audio recordings by the presence of active sound sources.
SED is typically posed as a supervised machine learning problem, requiring strong annotations for the presence or absence of each sound source at every time instant within the recording.
However, strong annotations of this type are both labor- and cost-intensive for human annotators to produce, which limits the practical scalability of SED methods.

In this work, we treat SED as a multiple instance learning (MIL) problem, where training labels are static over a short excerpt, indicating the presence or absence of sound sources but not their temporal locality.
The models, however, must still produce temporally dynamic predictions, which must be aggregated (pooled) when comparing against static labels during training.
To facilitate this aggregation, we develop a family of adaptive pooling operators---referred to as auto-pool---which smoothly interpolate between common pooling operators, such as min-, max-, or average-pooling, and automatically adapt to the characteristics of the sound sources in question.
We evaluate the proposed pooling operators on three datasets, and demonstrate that in each case, the proposed methods outperform non-adaptive pooling operators for static prediction, and nearly match the performance of models trained with strong, dynamic annotations.
The proposed method is evaluated in conjunction with convolutional neural networks, but can be readily applied to any differentiable model for time-series label prediction.
\rev{%
While this article focuses on SED applications, the proposed methods are general, and could be applied widely to MIL problems in any domain.}
\end{abstract}


%
\IEEEpeerreviewmaketitle%

\section{Introduction}
\IEEEPARstart{S}{ound} event detection (SED) is the task of automatically identifying the occurrence of specific sounds in continuous audio recordings.
Given a target set of sound sources of interest, \rev{the goal} is to return the start time, end time, and label (the class) of every sound event in the target set.
SED is a key component in a number of technologies and applications emerging from the recent advances in machine learning and Internet of Things (IoT) technology such as noise monitoring for smart cities~\cite{Bello:SONYC:CACM:18}, bioacoustic species and migration monitoring~\cite{Stowell:AEDoverlappinh:WASPAA:15,Salamon:FlightCalls:PLOSONE:16,Lostanlen:BirdVoxFullNight:ICASSP:18}, self-driving cars~\cite{DCASE17:SmartCars:Webpage:17}, surveillance~\cite{radhakrishnan:AudioSurveillance:WASPAA:05,Crocco:AudioSuerveillance:CS:16}, healthcare~\cite{Goetze:HealthSED:JCSE:12}, and large-scale multimedia indexing~\cite{Hershey:LargeAudioCNN:ICASSP:17}.

Modern SED systems are typically implemented by supervised machine learning algorithms, which are used to learn the parameters of a function to map a sequence of audio data to a sequence of event labels.
Because SED systems are required to produce \emph{dynamic} (time-varying) label estimates at each instant within a recording, they are often trained from \emph{strongly labeled training data}, where the presence or absence of each source at each instant is known.
While strongly labeled data is ideal for model development and evaluation, it is also costly to acquire.
As SED systems adopt data-intensive approaches---such as convolutional or recurrent neural networks---the availability and cost of strongly annotated data become serious impediments to system development.

If we are to accurately evaluate the dynamic performance of \rev{SED systems}, strongly labeled data is clearly necessary, and it is natural to assume that the same should hold for model development.
However, if one has access to a larger pool of data that has only been \emph{weakly labeled} at a coarse time resolution (\emph{e.g.}, 10~second clips), it may be possible to learn a high-quality, dynamic predictor with lower annotation costs.
The key to leveraging this kind of weakly labeled data lies in the means by which dynamic predictions are aggregated or \emph{pooled} across time to form static predictions.
There are standard approaches to aggregating predictions, such as \emph{max-} or \emph{mean-}pooling, which can be difficult to optimize (in the case of $\max$) or require strict assumptions about the characteristics of the data which may not hold in practice, \emph{e.g.}, mean-pooling assumes that event activation must occupy the majority of a labeled observation window.
Making effective use of weakly labeled data can therefore require substantial engineering and algorithm design effort.

\subsection{Our contributions}

In this article, we develop a general family of adaptive pooling operators---collectively referred to as \emph{auto-pool}---which generalize and interpolate between standard operators such as $\max$, $\text{mean}$, or $\min$.
The proposed methods are designed to be jointly learned with dynamic prediction models (\emph{e.g.}, convolutional networks), allowing dynamic predictors to be trained from weakly annotated data, and require minimal assumptions about the label characteristics.
We evaluate the proposed methods on three multi-label, sound event detection datasets, which exhibit differing characteristics of label sparsity and duration.
Our empirical results show that the proposed methods outperform standard, non-adaptive pooling operators, and the resulting models achieve comparable accuracy to models trained from strongly labeled data.

\section{Related work}
\subsection{Sound event detection}
Sound event detection (SED) has seen a dramatic increase in interest from the research community over the past decade, as evidenced by the growing popularity and participation in the DCASE challenge~\cite{mesaros2017dcase} and the emergence of domain-specific SED systems, \emph{e.g.},~for bioacoustic SED~\cite{Stowell:BioacousticCASA:Chapter:18}.
Early approaches relied on standard features (such as Mel-frequency cepstral coefficients) combined with standard machine learning algorithms, such as support vector machines~\cite{Temko:AED:PhD:07,Foggia:SED:PRL:15,Elizalde:SED:DCASE:16} or Gaussian mixture models (optionally with temporal smoothing)~\cite{cai:KeyAudioEffects:TASLP:06,Mesaros:AED:EUSIPCO:10,heittola:ContextEventDetection:EURASIP:13,Vuegen:AEDGMM:WASPAA:13}.
Other strategies include spectral decomposition methods and source separation models~\cite{Benetos:AED:ICASSP:16,Benetos:SEDPLCA:TASLP:17,heittola2013supervised,cotton:SpectroTemporal:WASPAA:11,Dikmen:SEDNMF:WASPAA:13,Gemmeke:AEDNMF:WASPAA:13,Mesaros:SEDNMF:ICASSP:15,Komatsu:AEDNMF:ICASSP:16}.
The most recent research on SED is dominated by deep (fully connected) neural networks~\cite{Cakir:SEDDNN:IJCNN:15}, convolutional networks~\cite{Cakir:FilterbankSED:IJCNN:16,Jeong:AEDCNN:DCASE:17,Lostanlen:BirdVoxFullNight:ICASSP:18}, recurrent networks~\cite{Parascandolo:RNNSED:ICASSP:16,Rui:SEDRNN:DCASE:17}, or convolutional-recurrent networks~\cite{Cakir:SEDCRNN:TASLP:17,Adavanne:SpatialSEDCRNN:ICASSP:17}.

The aforementioned approaches rely on strongly labeled data, which as discussed in the introduction, limits their practical \rev{applicability.} 
While this limitation can be overcome in part through data synthesis~\cite{salamon2017scaper}, the problem has led researchers to investigate models for SED that can be trained from weak (static) labels.
Interest in this problem formulation spiked with the release of AudioSet~\cite{Gemmeke:AudioSet:ICASSP:17}, which contains approximately 2 million 10-second YouTube clips with weak audio labels, and the DCASE~2017 challenge~\cite{mesaros2017dcase}.
One of the DCASE~2017 tasks (Task 4) was based on a subset of AudioSet, and the problem was to develop models that can be trained on weak labels but produce strong (\emph{i.e.}, dynamic, time-varying) labels.
Many of the subsequently published papers addressing SED using weakly labeled data (\cref{sec:sed_mil}) formulate the problem in the multiple instance learning (MIL) framework.

\subsection{Multiple instance learning}
Multiple instance learning (MIL) was proposed in its modern form by Dietterich et~al.~\cite{dietterich1997solving} as a supervised learning problem where a single binary class label is applied to a set (\emph{bag}) of related examples (instances) in the training set.
MIL problems naturally arise in a variety of application domains where precise labeling can be expensive, such as object recognition in computer vision.
A label may be applied to an image indicating the presence of an object, while the ``instances'' to be classified are small patches within the image.
Similarly, for SED, it may be more cost-effective to label a relatively long clip for the presence of an event, rather than each individual frame.

The general MIL formulation has been broadly applied within computer vision~\cite{zhang2006multiple,babenko2011multiple,hsu2014augmented}, it has been relatively less common in audio applications.
Mandel and Ellis~\cite{mandel2008multiple} compared two support vector machine-based MIL algorithms~\cite{andrews2003support,chen2006miles} for classifying 10-second music excerpts (the instances) for which labels had been generated at the levels of track, album, or artist.
Their target vocabulary included a mixture of genre, style, and instrumentation tags, and they found that the best-performing method was the MI-SVM algorithm~\cite{andrews2003support}, but that it was comparable to a naive baseline in which aggregated training labels were propagated to all constituent instances prior to training.
Relatedly, Wu~et~al.\ developed a hierarchical generative model for music emotion recognition~\cite{wu2014music}.
In their model, song labels are modeled as generating multiple instances of \emph{segments}, which in turn each generate multiple \emph{sentences} (instances) which are jointly represented by text (lyrics) and acoustic features.
While their generative model is trained on weakly labeled data, it does not provide a direct mechanism for inferring instance-level labels.

In other related work, Briggs~et~al.~\cite{briggs2012acoustic} compared several previously developed MIL algorithms for detecting (multiple) bird species from short (10--15s) audio excerpts.
Their results demonstrated that k-nearest-neighbor~\cite{zhang2005k} and clustering~\cite{zhou2007multi} approaches both perform well at excerpt-level prediction, but they did not report evaluations at the level of instances (time-frequency patches).
For comparison purposes, we evaluate the methods proposed here on both static and dynamic prediction.

\subsection{Sound event detection using weakly labeled data}
\label{sec:sed_mil}
When reviewing approaches for weakly labeled SED, we can group approaches by two key features: the model used to produce dynamic features (or predictions), \emph{i.e.}, an instance-level representation, and the approach used to aggregate instance-level features or predictions to a bag-level (static) prediction.
\rev{Note that for SED,} instances typically correspond to audio frames or short chunks.
In terms of modeling, while approaches based on GMM~\cite{Kumar:WeakSED:IJCNN:17} and SVM~\cite{Kumar:WeakAED:ACMMM:16} have been proposed, the vast majority are based on deep neural networks including DNN~\cite{Kong:WeakSED:ICASSP:17}, CNN~\cite{Su:WeakAED:ICASSP:17,Chou:FrameCNN:DCASE:17,Salamon:MILSED:DCASE:17,Kumar:WeakSED:ICASSP:18,Xu:WeakSED:ICASSP:18}, RNN~\cite{Wang:WeakSED:ICASSP:17} and CRNN~\cite{Adavanne:WeakSED:DCASE:17}.
Some approaches propagate the bag-level label to all instances and train against these directly~\cite{Chou:FrameCNN:DCASE:17,Adavanne:WeakSED:DCASE:17}, which can introduce instance-level label noise.
Other approaches are based on source separation, and obtain dynamic labels by post-processing the separated sources (\emph{e.g.}, by computing the frame-wise energy of each separated source)~\cite{Kong:WeakSED:ICASSP:18,Sobieraj:WeakSED:ICASSP:18}.

However, the majority of approaches aggregate instance-level representations over time to produce a bag-level prediction.
Given the standard MIL formulation, it is understandable that most approaches rely on pooling or customized loss functions that make use of the $\max$ operator~\cite{Kumar:WeakAED:ACMMM:16,Su:WeakAED:ICASSP:17,Kumar:WeakSED:IJCNN:17}, though variants including (a precursor to this work) soft-max pooling~\cite{Salamon:MILSED:DCASE:17}, and mean pooling~\cite{Kumar:WeakSED:ICASSP:18} have been proposed.
As shall be discussed in \Cref{sec:methods}, max-pooling causes a number of issues that limit its efficacy as a pooling strategy for MIL.\@

\rev{\subsection{Attention and differentiable pooling}}

Attention mechanisms~\cite{bahdanau2014neural} have been recently developed as a way to restrict the dependence of an output prediction to a subset of the input.
Typically, attention mechanisms are applied to \emph{structured prediction} problems, such as machine translation or automatic speech recognition, where the output is a sequence (\emph{e.g.}, predicted translation text) has some regular structure that may be exploited by the model architecture, which is often a recurrent neural network.
While the basic idea of attention for MIL is appealing, the training labels in MIL are typically \emph{unstructured}: \emph{e.g.}, a single label that applies to an entire sequence.
However, the model must still produce structured predictions, and it is not directly obvious how to apply standard attention mechanisms.

Convolutional (feed-forward) attention~\cite{raffel2015feed} is a closer fit to the MIL setting, as the attention mechanism is used to summarize a structured input by a fixed-length \emph{context vector} $c$ as a weighted average $c = \sum_t e_t h_t$ of instance representations $\{h_t\}$, from which the output is predicted as $\hat{y} = {g( c )}$.
\rev{Note that the intermediate representations $h_t$ do not generally constitute instance predictions.
    Because $g$ is usually non-linear (and non-convex), there is no direct relationship between the attention-aggregated output ${g( c )} = g(\sum_t e_t h_t)$, the weighted average of $g$ applied to instances $\sum_t e_t g(h_t)$, and instance-wise outputs $g(h_t)$.
As a result, while optimizing $g( c )$ may provide a good bag-level model, it does not directly provide an instance-level model as required by MIL.}

Recently, attention-based models have been proposed which use class likelihoods as intermediate representations, along with an identity mapping for $g$~\cite{Kong:WeakSED:ICASSP:17,Xu:WeakSED:ICASSP:18}.
The methods we develop in this article are similar in spirit, but with a more constrained and interpretable attention mechanism that relates directly to the instance-level predictions and the MIL problem formulation.
Moreover, the proposed methods introduce only a single additional parameter for each class, which can be directly interpreted as interpolating between different standard pooling operators, as described below.

\rev{%
Aside from attention models, similar techniques have been used to provide differentiable pooling operators for MIL.\@
Most similar to the methods we propose is that of Hsu et~al.~\cite{hsu2014augmented}, which uses the smooth $\log \sum \exp$ approximation to the $\max$ operator for MIL applications in computer vision.
Hsu~et~al.\ introduce a hyper-parameter to control the sharpness of the approximation, but it is fixed a priori, and unlike the methods proposed here, the aggregation does not adapt automatically.
Moreover, because $\log\sum\exp$ aggregation is non-linear, it exhibits similar difficulty in recovering instance-wise predictions as the attention-based approach described above.

Finally, Zeiler and Fergus~\cite{zeiler2012differentiable} developed a general formulation that adaptively interpolates between different standard pooling operators.
Although this approach has been applied to audio problems~\cite{swietojanski2015differentiable}, its use has been limited to pooling of internal feature representations in convolutional networks, and it has not been used in MIL applications.
The methods we develop in this article are conceptually simpler, and more limited in scope to directly address the difficulties of aggregating predictions in MIL problems.

}

\section{Methods}
\label{sec:methods}
In this section, we describe the multiple instance learning problem in general, and illustrate short-comings of standard pooling operators when applied to the MIL context.
We then develop a family of adaptive pooling operators to reduce dynamic label predictions to static predictions during training.
\rev{For ease of exposition, we first derive the methods for single-label binary classification problems, followed by the generalization to multi-label settings.}

\subsection{Multiple instance learning}
In the multiple instance learning (MIL) problem formulation, training data are provided as labeled sets (\emph{bags}) of examples ${(X_i, Y_i)}_{i=1}^n$ where $X_i = \{x_1, x_2, \cdots\} \subset \mathcal{X}$ contains multiple \emph{instances} $x_j$, and $Y_i \in \{0, 1\}$ is a single label for the set $X_i$~\cite{dietterich1997solving}, and $\mathcal{X}$ and $\mathcal{Y}$ denote the feature and label spaces.
The label convention is that $Y_i = 1$ if any instance $x \in X_i$ is positive, and $Y_i = 0$ if all instances are negative.
The goal is to use this weakly labeled data to learn an instance classifier $h : \mathcal{X} \rightarrow \mathcal{Y}$.

While MIL can be applied to a variety of learning algorithms (\emph{e.g.}, support vector machines or nearest neighbor classifiers), in this work we focus on deep neural networks.
The classifiers under consideration here take the form of a thresholded likelihood $\hat{p}(Y\given x)$, \emph{e.g.},
\begin{equation}
h(x) = \begin{cases}
			1 & \hat{p}(Y \given x) \geq 0.5\\
			0 & \text{otherwise}
		\end{cases}.
\end{equation}
In this formulation, the predicted label for a bag is the maximum over instance-wise predictions.
Equivalently, a likelihood for the bag-label can be induced from the instance likelihoods by defining the bag-level likelihood as
\begin{equation}
	\hat{P}(Y \given X) = \max_{x \in X} \hat{p}(Y\given x),\label{eq:milprob}
\end{equation}
which results in the bag prediction rule
\begin{equation}
\overline{h}(X) = \begin{cases}
			1 & \hat{P}(Y\given X) \geq 0.5\\
			0 & \text{otherwise}
        \end{cases}.\label{eq:bagpredict}
\end{equation}
This prediction rule is depicted schematically in \Cref{fig:milgrads} (left), where the bag-level prediction depends only on the maximum of its instance-level predictions.
\begin{figure}
\centering
\includegraphics[width=\columnwidth]{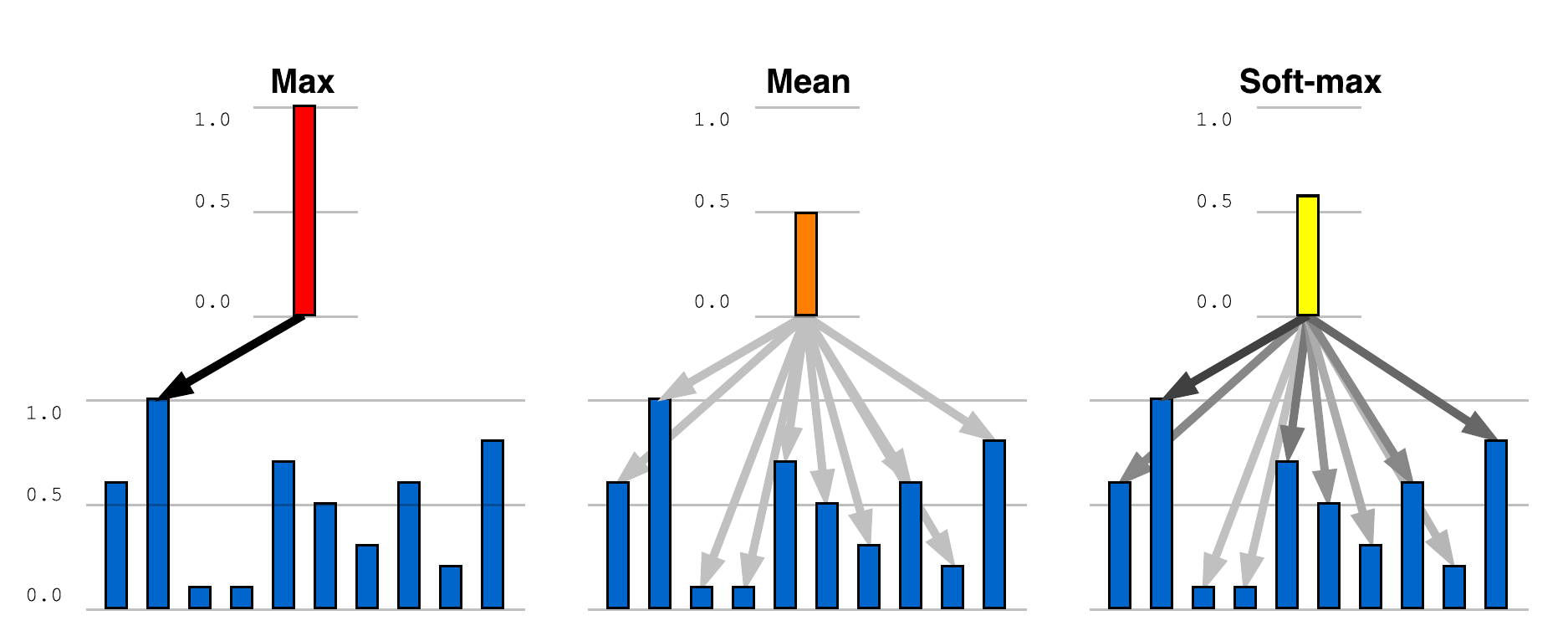}
\caption{Pooling operators propagate gradient information in proportion to the responsibilities they assign to instance-level predictions, indicated here by the darkness of arrows.
Left: $\max$ assigns all responsibility to the largest instance; middle: mean assigns equal responsibility to all instances; right: soft-max (\cref{eq:smp}) assigns greater responsibility to large instances.\label{fig:milgrads}}
\end{figure}

During training, the objective is to maximize the likelihood of observed labeled bags, \emph{e.g.}, by minimizing the binary cross-entropy over the model parameters $\theta$:
\begin{equation}
\min_\theta \frac{1}{n}\sum_{i=1}^n -Y_i \log \hat{P}(Y\given X_i) - (1 - Y_i)\log\left(1 - \hat{P}(Y\given X_i)\right).\label{eq:milopt}
\end{equation}

\subsection{Max-pooling}
Typically, \cref{eq:milopt} is optimized by some form of gradient descent, which requires propagating gradients through the $\max$ operator in \cref{eq:milprob} via the chain rule.
The $\max$ operator is not itself differentiable, so sub-gradient descent must be used instead.
The \emph{sub-differential set} of the $\max$ operator applied to inputs $\{z_i\}\subset \mathbb{R}$ is the set of all convex combinations of its maximizers's \emph{sub-gradients} (assuming each input is sub-differentiable):
\begin{equation}
    \partial \max\left\{z_i \right\} = \text{Conv}\left\{g \mgiven g \in \partial z_i \;\wedge\; z_i = \max_j z_j \right\}
\end{equation}
where $\text{Conv}(S)$ denotes the convex hull of a set $S$:
\begin{equation*}
\text{Conv}(S) = \left\{\sum_{x \in S} \mu_x x \mgiven \sum_{x \in S}\mu_x = 1\;\wedge \; \forall_x \mu_x \geq 0~\right\}.
\end{equation*}
Any element of $\partial\max$ can be used in place of a gradient, though most implementations select a single maximizer at random; often, the maximizer is unique, so the distinction is unimportant.
A sub-gradient of $\max$ can be thus viewed as a weighted average of all inputs, subject to the constraint that non-maximizing inputs must have weight 0.

\rev{When applying the chain rule to $\partial\max$, the sub-gradient of the objective function with respect to non-maximizing instances is 0, and those instances therefore do not contribute when updating the parameters $\theta$.}
This is particularly problematic early in training, where the instance-wise predictions are essentially random.
Parameter updates then depend entirely upon single, randomly selected instances (as depicted in \Cref{fig:milgrads}, left).
As a result, max-pooling for MIL can be sensitive to initialization, generally unstable, and \rev{difficult to deploy}.

\subsection{Soft-max pooling}
To ameliorate the issues highlighted above, we proposed in previous work~\cite{Salamon:MILSED:DCASE:17} to replace the $\max$ operator in \cref{eq:milprob} by the \emph{soft-max weighted average}:
\begin{equation}
	\hat{P}_s(Y\given X) = \sum_{x \in X} \hat{p}(Y\given x) \left(\frac{\exp \hat{p}(Y\given x)}{\displaystyle\sum_{z\in X}\exp \hat{p}(Y\given z)}\right).\label{eq:smp}
\end{equation}
This operator behaves similarly to the $\max$ operator, in that $\hat{P}_s$ is large if any of its inputs $\hat{p}(Y\given x)$ are large, and small if all of its inputs are small.
However, it is continuously differentiable, and assigns responsibility to each instance $x$ so that the entire bag contributes to the gradient calculation and parameter updates.
As illustrated in \Cref{fig:milgrads} (right), each instance $x$ contributes in proportion to its label likelihood $\hat{p}(Y\given x)$, so that positive predictions have more influence and negative predictions have less.

Because the inputs to the soft-max pooling operator are probabilities $\hat{p}(Y\given x) \in [0, 1]$, the weights assigned by \cref{eq:smp} are also bounded.
In general, we have the following relation between a soft-max's input and output:
\begin{proposition}
    Let $a \leq b \in \mathbb{R}$ and $z \in {[a, b]}^m \subset \mathbb{R}^m$, and let $\rho{(z)}_i := \exp(z_i) / \sum_j \exp(z_j)$ denote the soft-max operator.
    Then for any coordinate $i$, the corresponding soft-max output $\rho{(z)}_i$ is bounded as
\[
    \frac{\e^a}{\e^a + (m-1) \cdot \e^b} \leq \rho{(z)}_i \leq \frac{\e^b}{\e^b + (m-1) \cdot \e^a}.
\]\label{smbound}
\end{proposition}
\begin{proof}
    First, observe that $\rho{(z)}_i$ is proportional to $\exp z_i$ and inversely proportional to $\sum_{j\neq i} \exp z_j$.
    A soft-max coordinate $\rho{(z)}_i$ is therefore maximal when one coordinate $z_i = b$ is maximal, and all remaining $(m-1)$ coordinates $z_{j\neq i} = a$ are minimal.
In this case, the soft-max output $\rho(z)$ \rev{for each coordinate $k$} is
\begin{eqnarray*}
    \rho{(z)}_k &=& \frac{\exp(z_k)}{\e^b + (m - 1) \cdot \e^a}.
\end{eqnarray*}
Since all $z_k \leq b$, this achieves the upper bound.
A similar argument proves the analogous lower bound.
\end{proof}

Applying \cref{smbound}, if a bag has $|X|=m$ instances, and each instance $x \in X$ has a likelihood $0 \leq \hat{p}(Y\given x) \leq 1$, then the weight for each instance is bounded as
\begin{equation}
\frac{1}{1 + \e \cdot (m-1)}
	\leq  \frac{\exp \hat{p}(Y\given x)}{%
    	\displaystyle\sum_{z\in X}\exp\hat{p}(Y\given z)}
	\leq \frac{\e}{\e + m-1}.\label{eq:smpbound}
\end{equation}
Soft-max pooling therefore has limited capacity to concentrate on a small portion of instances within a bag, since the weight for any single instance is $\Theta(1/m)$.
As illustrated in \Cref{fig:smpbounds}, soft-max pooling behaves similarly to unweighted averaging as the bag size grows.

\begin{figure}
    \centering%
    \includegraphics[width=\columnwidth]{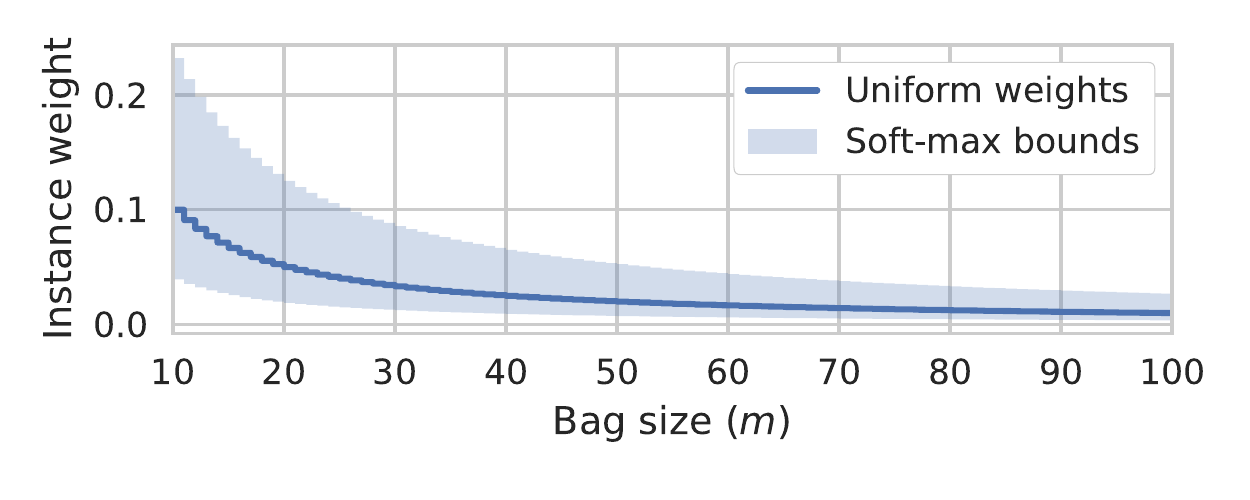}
    \caption{The soft-max weighted average (\cref{eq:smp}) produces instance weights satisfying the bounds given in \cref{eq:smpbound} (shaded region).
    As the size $m$ of the bag grows, the bounds converge to $1/m$ (solid line).\label{fig:smpbounds}}
\end{figure}

\subsection{Auto-pooling}

The bounded range problem of soft-max pooling can be addressed by introducing a scalar parameter $\alpha \in \mathbb{R}$:
\begin{equation}
	\hat{P}_\alpha(Y\given X) = \sum_{x \in X} \hat{p}(Y\given x)  \left(\frac{\exp\left(\alpha\cdot\hat{p}(Y\given x)\right)}{\displaystyle\sum_{z\in X}\exp\left(\alpha\cdot\hat{p}(Y\given z)\right)}\right).\label{eq:autopool}
\end{equation}
Treating $\alpha$ as a free parameter to be learned along-side the model parameters $\theta$ allows \cref{eq:autopool} to automatically adapt to and interpolate between different pooling behaviors.
For example, when $\alpha=0$, \cref{eq:autopool} reduces to an unweighted \rev{mean} (\Cref{fig:milgrads}, center); when ${\alpha=1}$, \cref{eq:autopool} simplifies to soft-max pooling \cref{eq:smp}; and when ${\alpha\rightarrow\infty}$, \cref{eq:autopool} approaches the max operator.
We therefore refer to the operator in \cref{eq:autopool} as \emph{auto-pooling}.

With auto-pooling, for $\alpha \geq 0$, the bounds from \cref{smbound} are $[a, b] = [0, \alpha]$, and the instance weights are bounded by
\begin{equation}
\frac{1}{1 + \e^\alpha \cdot (m-1)}
	\leq  \frac{\exp\left(\alpha\cdot\hat{p}(Y\given x)\right)}{%
    	\displaystyle\sum_{z\in X}\exp\left(\alpha\cdot\hat{p}(Y\given z)\right)}
	\leq \frac{\e^\alpha}{\e^\alpha + m - 1},\label{eq:apbound}
\end{equation}
which approaches the open unit interval $(0, 1)$ as $\alpha\rightarrow\infty$.

Additionally, letting $\alpha \leq 0$ leads to approximate \emph{min}-pooling, where smaller input values receive larger weight in the combination.
In this case, the bounds are $[a, b] = [\alpha, 0]$, and the resulting instance weight bounds are:
\begin{equation}
\frac{\e^\alpha}{\e^\alpha + m-1}
	\leq  \frac{\exp\left(\alpha\cdot\hat{p}(Y\given x)\right)}{%
    	\displaystyle\sum_{z\in X}\exp\left(\alpha\cdot\hat{p}(Y\given z)\right)}
	\leq \frac{1}{1 + \e^\alpha \cdot ( m - 1)}.\label{eq:aplbound}
\end{equation}
As $\alpha\rightarrow-\infty$, the weight bounds again approach the unit interval $(0, 1)$, except that the upper bound is now achieved by the \emph{smallest} instance prediction.
This effectively relaxes the core assumption of multiple instance learning that a bag label is equal to the max (disjunction) over instance labels.
\rev{Supporting min (conjunction) behavior allows for a bag to be predicted as a positive example if \emph{all} of its instances are predicted as positive examples, which would be appropriate for long-duration events.}

\subsection{Constrained auto-pooling}

\Cref{eq:apbound} bounds the effective range of the weight assigned to any given instance in terms of the pooling parameter $\alpha$ and the bag size $m$.
However, in some applications, it may be more natural to constrain $\alpha$ in terms of the amount of weight the pooling operator is allowed to assign to a single instance when making a bag-level decision.
For example, in sound event detection, this may correspond to requiring the detector to be active for at least some minimum time duration for the bag to be predicted as a positive example.
\rev{Alternately, one may require that a minimum fraction of \emph{instances} must be positive before the bag is predicted positive, or equivalently, that no single instance receives too much weight in~\eqref{eq:autopool}.}

Let $1/m \leq \phi_+ < 1$ denote the maximum permissible aggregation weight for a single instance.%
\footnote{The maximum weight $\phi_+$ cannot be less than $1/m$ because all weights must sum to 1.
Similarly, a minimum weight bound $\phi_-$ cannot exceed $1/m$.}
Then $\alpha\geq 0$ can be upper-bounded as:
\begin{equation}
\frac{\e^\alpha}{\e^\alpha + m - 1} \leq \phi_+\quad\Rightarrow\quad
\alpha \leq \ln \frac{\phi_+}{1-\phi_+} + \ln\left(m -1\right).\label{alphabound}
\end{equation}
Similarly, a minimum weight constraint $0 < \phi_- \leq 1/m$ produces the following lower bound for $\alpha$:
\begin{equation}
\phi_- \leq \frac{\e^\alpha}{\e^\alpha + m-1}\quad\Rightarrow\quad
\alpha \geq \ln \frac{\phi_-}{1-\phi_-} + \ln\left(m -1\right).
\end{equation}
Note that these bounds are tight, in that $\phi_- = \phi_+ = 1/m$ implies $\alpha = 0$, which recovers mean-pooling.

Now, consider the minimal upper bound $\phi_+$ that allows a single instance to determine the majority vote for a bag.
This is achieved by the extremal case where one instance $i$ is maximal and the remaining instances $j\neq i$ are minimal:
\begin{equation}
\hat{p}\left(Y\mgiven x_k\right) = \begin{cases}
1 & k = i\\
0 & k \neq i
\end{cases}.
\end{equation}
With the decision rule (and threshold) given in \cref{eq:bagpredict}, ${\phi_+ = 0.5}$ is the minimal upper bound on weights that produces max-pooling behavior, and therefore constitutes an upper bound that does not significantly reduce the flexibility of auto-pooling.
With this value of $\phi_+$, \cref{alphabound} simplifies to
\[
\phi_+ = 0.5 \quad \Rightarrow \quad \alpha \leq \ln(m-1).
\]
Throughout the remainder of this article, we will refer to auto-pooling with the $\phi_+=0.5$ bound imposed as \emph{constrained auto-pool (CAP)}.

\subsection{Regularized auto-pooling}
\rev{%
As an alternative to constrained auto-pool, one may consider \emph{regularized auto-pool (RAP)}, where a penalty is applied to $\alpha$ to prevent it from placing too much weight on individual instances but without an explicit bound on the maximum (or minimum) weight.
While there are many possibilities for the choice of penalty function, here we opt for a quadratic penalty $\alpha^2$, so that the penalty grows with $\alpha$.
This promotes mean-like behavior, but still provides flexibility to learn max-pooling behavior if necessary.

Concretely, for the remainder of this article, we will denote by \emph{RAP} any auto-pool model with a quadratic penalty:
\[
\min_{\theta,\alpha} f(\theta) + \lambda |\alpha|^2,
\]
where $f(\theta)$ denotes the learning objective of \cref{eq:milopt},
and ${\lambda>0}$ is a positive coefficient.
For multi-label formulations, the penalty generalizes to the squared Euclidean norm $\lambda \|\alpha\|^2$.}

\subsection{Multi-label learning}
The discussion so far has centered on binary classification problems, but the methods directly generalize to multi-label settings, in which each instance $x$ receives multiple positive labels.
In this setting, a separate auto-pooling operator is applied to each class.
Rather than a single parameter $\alpha$, there is a vector of parameters $\alpha_c$ where $c$ indexes the output vocabulary.
This allows a jointly trained model to adapt the pooling strategies independently for each category.

\section{Experiments}

In this section, we describe a series of experiments investigating the behavior of auto-pooling methods on three sound event detection applications: urban environments \rev{(URBAN-SED)}, smart cars \rev{(DCASE~2017)}, and musical instruments \rev{(MedleyDB)}.
For each dataset, we compare models trained with standard, non-adaptive pooling operators ($\max$ and mean), the soft-max pooling model described in \Cref{sec:methods}, and the three adaptive methods: auto-pool, constrained auto-pool (CAP), and regularized auto-pool (RAP).
For RAP models, we report results independently for $\lambda \in \{10^{-2}, 10^{-3}, 10^{-4}\}$.
For the urban environment and musical instrument applications, we will also compare to a model trained with strong (time-varying) labels to provide a sense of the maximum expected performance for the given model architecture.
\rev{Models trained with strong labels omit the temporal pooling step, and the training loss is computed independently for each instance.}
The smart car dataset \rev{(DCASE~2017)} does not provide strong labels for the training set, so this comparison could not be performed.

We report standard evaluation metrics for both static (bag-level) and dynamic (segment-level) prediction: precision, recall, and $F_1$.\footnote{$F_1$-macro \rev{reports the unweighted average of class-wise $F_1$ scores.} Micro-averages are not available for segment-based evaluation.}
For static predictions the metrics are computed following the standard methodology for multi-label classification evaluation. 
For the dynamic predictions we compute the segment-based variant of the aforementioned metrics as defined by Mesaros et al.~\cite{Mesaros_2016} using a segment duration of $1$~second, as per the DCASE challenge~\cite{mesaros2017dcase}.
Additionally, for the dynamic prediction task, we report the \emph{error rate} $E$, defined as the average number of substitutions, insertions, and deletions of events over all segments.
Note that precision, recall and $F_1$ range from 0 (worst) to 1 (best), while the error rate $E$ is non-negative with 0 being the best and greater values representing worse performance.

\subsection{Datasets}

\subsubsection{URBAN-SED}
URBAN-SED is a dataset of 10,000 soundscapes generated using the Scaper soundscape synthesis library~\cite{salamon2017scaper}.
Each soundscape has a duration of 10~s, and the dataset as a whole totals 27.8 hours of audio with close to 50,000 annotated sound events from 10 sound classes.
Each soundscape contains between 1--9 foreground sound events, where the source material for the events comes from the UrbanSound8K dataset~\cite{Salamon:UrbanSound:ACMMM:14}, and has a background of Brownian noise resembling the typical ``hum'' often heard in urban environments.
The dataset comes pre-sorted into train, validation and test splits containing 6000, 2000 and 2000 soundscapes respectively. 

An important characteristic of URBAN-SED is that since both the audio and annotations were generated computationally, the annotations are guaranteed to be correct and complete, while the dataset is an order of magnitude larger than the largest strongly labeled SED dataset compiled via manual labeling.
Since the soundscapes are ``composed'' using a process akin to an audio sequencer, they are not as realistic as manually labeled datasets of real soundscape recordings.
In particular, as illustrated in \Cref{fig:urbansed:durations}, events may be artificially truncated in duration in unnatural-sounding ways.
Still, it has been shown that the data still present a challenging scenario for state-of-the-art SED models~\cite{salamon2017scaper}.

\begin{figure}
\centering
\includegraphics[width=\columnwidth]{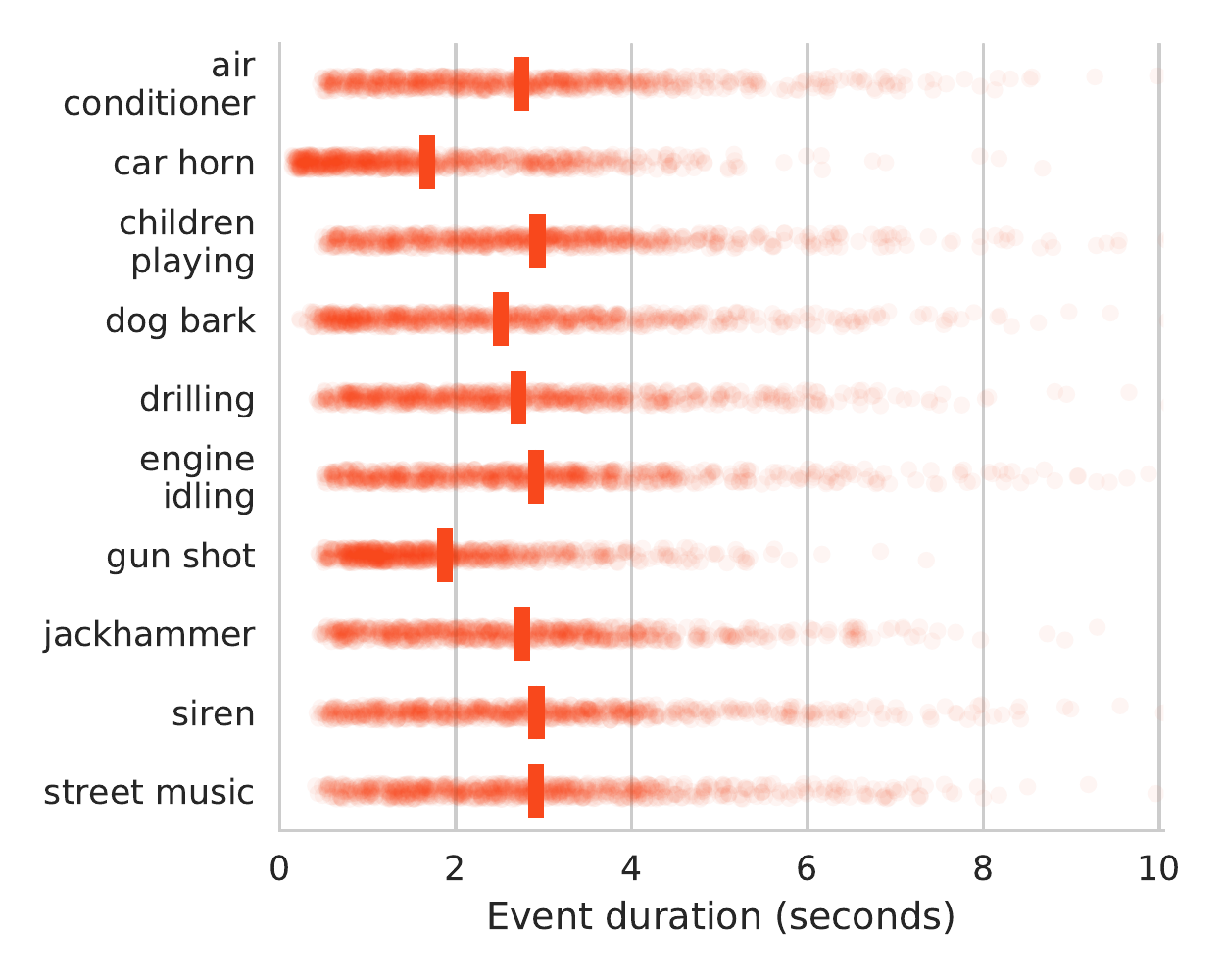}%
\caption{Event durations for each class in the URBAN-SED.\@
    Each point corresponds to a test clip, and the mean event durations are indicated by vertical bars.
By construction, each event is clipped to at most 3~s (30\% of the clip), though an event class can occur multiple times within a clip.\label{fig:urbansed:durations}}
\end{figure}

\subsubsection{DCASE~2017 Task 4} 
The DCASE~2017 challenge~\cite{mesaros2017dcase} consisted of four tasks, including one task with the same problem formulation as this work (training a model to generate strong predictions from weakly labeled training data), task 4: ``Large-scale weakly supervised sound event detection for smart cars''.
The dataset used for this task is a subset of the AudioSet dataset~\cite{Gemmeke:AudioSet:ICASSP:17}, and consists of just over 50K 10-second excerpts from YouTube videos.
The dataset is split into a ``development'' set of 51660 excerpts and an ``evaluation'' (\emph{i.e.}, test) set of 1103 excerpts.
The development set is further divided into a ``train'' set with 51172 excerpts and a validation set of 488 videos.\footnote{To avoid possible confusion it is necessary to highlight the difference between the nomenclature used in the challenge and the nomenclature more commonly found the literature, as the latter will be used in this study for consistency. Throughout this paper, we use the challenge ``train'' set as our training set, the challenge ``test'' set as our validation set, and the challenge ``evaluation'' set as our test set.}

For conciseness, for the remainder of the paper we shall refer to this dataset simply as ``DCASE~2017''.
The sound events in this dataset come from 17 sound classes selected by the challenge organizers out of the AudioSet ontology~\cite{Gemmeke:AudioSet:ICASSP:17} that are related to traffic such as sirens, horns, beeps, and different types of vehicles such as car, bus and truck.
The weak labels were generated semi-automatically~\cite{Gemmeke:AudioSet:ICASSP:17}, while strong labels for the validation and test sets were manually annotated by the challenge organizers by listening to the audio (without watching the video).
The dataset fits our problem formulation, but its annotations have limitations, which makes proper evaluation difficult. 
Not all target sound events are guaranteed to be labeled and have a non-zero duration, and some such as ``car'' and ``car passing by'' are semantically overlapping.
However, it does have a unique distribution of event durations compared to the other two datasets used in this study.
The event durations for this dataset, depicted in \Cref{fig:dcase:durations}, follow a more natural distribution than those of URBAN-SED (\Cref{fig:urbansed:durations}), which we expect to influence the behavior of the proposed models, in particular the auto-pool models where $\alpha$ is learned from the data.
\rev{Our motivation for including this dataset in the evaluation is primarily to study the adaptive behavior of the $\alpha$ parameter, and not to achieve the best possible performance (measured by $F_1$).\footnote{For a thorough evaluation of existing methods on this dataset, we refer interested readers to the DCASE~2017 challenge results: \url{https://www.cs.tut.fi/sgn/arg/dcase2017/challenge/task-large-scale-sound-event-detection-results}.}}

\begin{figure}
    \centering%
    \includegraphics[width=\columnwidth]{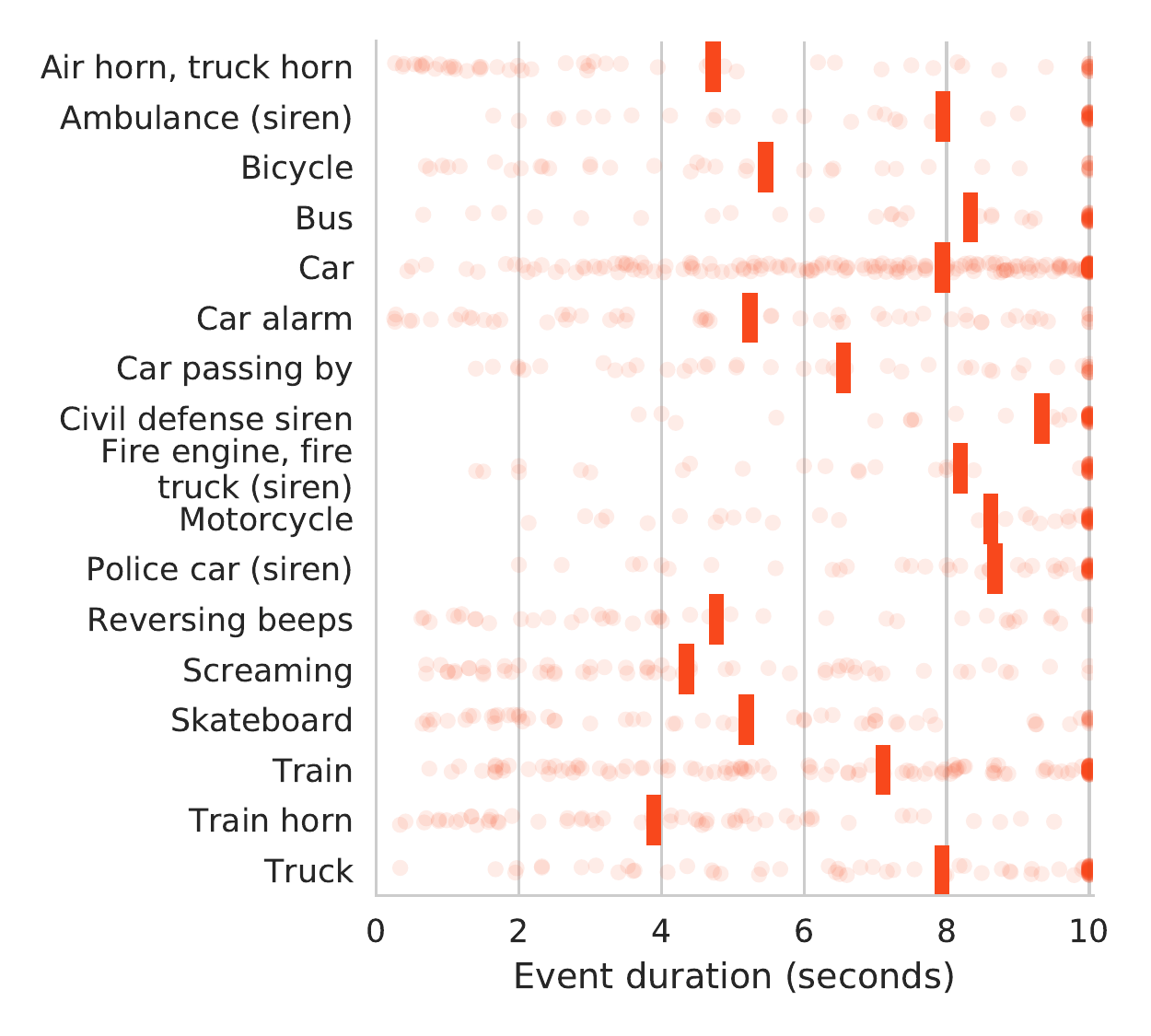}
    \caption{Event durations for each class in DCASE~2017.
    Each point corresponds to a test clip, and the mean event durations are indicated by vertical bars.
DCASE events typically cover at least 40\% (4~s) of the clip, and the high concentrations at 10.0 indicate that events often span the entire clip.\label{fig:dcase:durations}}
\end{figure}

\subsubsection{MedleyDB}
MedleyDB~\cite{bittner2014medleydb} is a collection of 122 multi-track recordings from a variety of musical genres and styles.
While it was initially developed to facilitate pitch tracking evaluation, it includes time-varying instrument activation labels for each track.

Because each track in MedleyDB is provided in the form of isolated instrument recordings (\emph{stems}), it is possible to generate different mixtures of the stem recordings for any given track.
This motivates a form of data augmentation: if a track has $n$ instruments, we generate $n$ alternate mixes, where mix $i$ has the $i$th instrument removed; the remaining $n-1$ stems are linearly mixed to best approximate the full mix, using the mixing coefficients provided by the MedleyDB python package.\footnote{\url{https://github.com/marl/medleydb}}
By training on this expanded set of \emph{leave-one-out} mixes, we separate each instrument from its surrounding context, which helps to eliminate confounding factors when estimating the presence of each instrument.
The expanded MedleyDB set contains 531 tracks, totaling 33.1 hours of audio.

Because of the skewed distribution of instruments in MedleyDB, we reduced the vocabulary of interest to the 8 most common sources: \emph{acoustic guitar, clean electric guitar, distorted electric guitar, drum set, electric bass, female singer, male singer, piano}.
Unlike URBAN-SED and DCASE, there is not a pre-defined evaluation split of MedleyDB.\@
We instead repeated the experiment over 10 random, artist-conditional 80--20 train-test splits; validation sets were randomly split 80--20 from the training splits (without artist conditioning).
Having multiple train-test splits allows us to perform statistical analyses which are not possible with the URBAN-SED and DCASE datasets.
\rev{We therefore do not make claims as to which methods perform ``best'' on URBAN-SED and DCASE.}

\Cref{fig:medley:durations} illustrates the distribution of instrument activation durations over the dataset.
Most instruments are active for substantially longer than the 10~s observation window used in our experiments, indicating that labels should be expected to be constant (entirely on or entirely off) over the duration of a training example.

\begin{figure}
\centering
\includegraphics[width=\columnwidth]{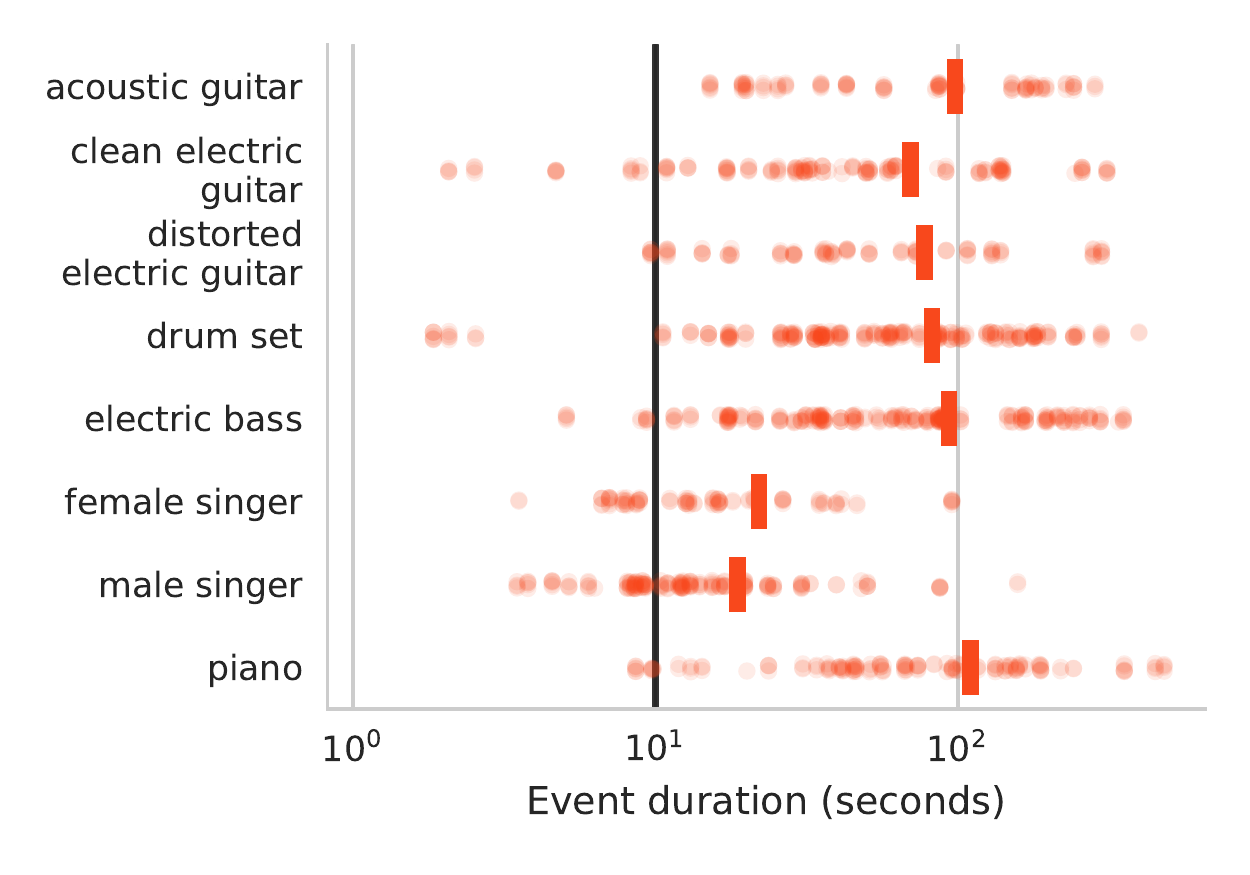}%
\caption{Event durations for each class in MedleyDB (logarithmically scaled).
Each point corresponds to the total duration of an instrument over a track, with the mean durations indicated by vertical bars.
The black line marks the 10~second point used to generate training patches.\label{fig:medley:durations}}
\end{figure}

\subsection{Model architecture}
The model used in this work is divided into two main components: a \emph{dynamic} predictor that generates predictions at a fine temporal resolution (\emph{i.e.}, frame/instance-level predictions), and a pooling layer which aggregates the instance-level predictions into a single \emph{static} (bag-level) prediction.
Our goal is to compare and contrast the different pooling functions proposed in \Cref{sec:methods}.
As such, in this work we adopt a single model architecture for the dynamic predictor, and keep it fixed throughout the study.
A block diagram depicting the complete architecture including the dynamic predictor followed by the temporal pooling layer is provided in \Cref{fig:blockdiagram}.

\begin{figure*}
\centering
\includegraphics[width=0.85\textwidth]{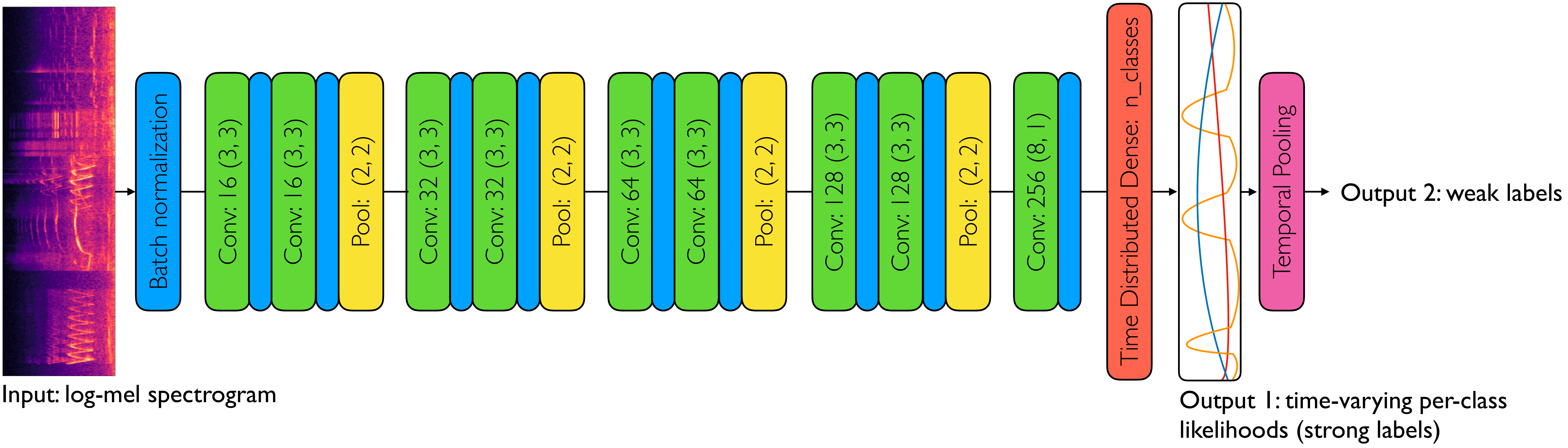}
\caption{Block diagram of the model architecture used in this study, including its two main components: a fully convolutional dynamic predictor, followed by a temporal pooling layer implemented by any of the pooling functions described in \Cref{sec:methods}: max, mean, soft-max, auto-pool, constrained auto-pool (CAP) or regularized auto-pool (RAP).}
\label{fig:blockdiagram}
\end{figure*}

For the dynamic predictor, we use an architecture inspired by the audio subnetwork of the $L^3$-Net architecture proposed by Arandjelovic and Zisserman~\cite{arandjelovic2017look}, which was shown to learn highly discriminative deep audio embeddings from a self-supervised audio-visual correspondence task.
In this work the input dimensions are ordered as (feature, time); details about the input are provided in \Cref{sec:experiments_traineval}.
The model begins with four convolutional blocks, each block consisting of two convolutional layers followed by strided $(2,2)$ max-pooling, where the number of convolutional filters is doubled for each subsequent block $(16, 32, 64, 128)$ and all filters are of dimensionality $(3,3)$.
This is followed by a single convolutional layer with 256 full-height $(8,1)$ filters, followed by a \rev{single dense layer applied independently to each time-step,} and with as many outputs (units) as there are classes in the dataset being used.
Batch normalization~\cite{Ioffe:BatchNorm:ICML:15} is applied to the output of every convolutional layer as well as to the input to the network.
We apply dimensionality-maintaining padding (``same padding'') to the input to all convolutional layers but the last, where we do not apply padding (``valid padding'').
We use rectified linear unit (ReLU) activations for all convolutional layers and sigmoid activations for the \rev{output} layer to support multi-label classification.
The output of the latter is a multi-label prediction $\hat{p}(Y\given x)$ for each frame (instance) $x$. 
Note that since the model down-samples in time by max-pooling, the frame rate of the dynamic predictions is reduced by a factor of 16 from the input.

Finally, the output of the dynamic predictor is aggregated over all instances using one of the pooling operators presented in \Cref{sec:methods} to produce a static prediction $\hat{P}(Y\given X)$ for each class, represented in \Cref{fig:blockdiagram} by the temporal pooling layer at the right end of the diagram.

Note that since the dynamic predictor is composed of convolutional layers and a single time-distributed dense layer, it is agnostic of the input length (\emph{i.e.}, the number of input instances/frames).
This is followed by the temporal pooling layer which is again agnostic of the input length.
As such, the entire architecture is length-agnostic (for audio, duration-agnostic) and can accept input of arbitrary length.
That said, some of the pooling functions presented in \Cref{sec:methods} \emph{are} affected by the length of the input: \emph{e.g.}, soft-max pooling approaches mean pooling as the length of the input increases, and the parameter $\alpha$ for auto-pooling methods depends on the bag length $m$.
However, this only matters for static prediction, and after the model has been trained, it can still produce dynamic predictions on arbitrary-length inputs.

\subsection{Training and evaluation}
\label{sec:experiments_traineval}
In all experiments, training data was augmented using MUDA~\cite{mcfee2015software} to generate pitch-shifted versions of each example by $\{\pm 1, \pm 2\}$ semitones, increasing the effective training set size by a factor of 5.
All signals were processed with librosa 0.5.1~\cite{mcfee_brian_2017_293021} to produce log-scaled Mel spectrograms with the following parameters: sampling rate 44.1~KHz, $n_\text{FFT} = 2048$ (46ms windows), hop length of 1024 samples (frame rate of 43~Hz), and 128 Mel frequency bands.
The models produce dynamic predictions at a frame rate of $43/16\approx2.69$~Hz.

Models were implemented using Keras~\cite{chollet2015keras} and TensorFlow~\cite{abadi2016tensorflow}.
Each model was trained using the Adam optimizer~\cite{kingma2014adam}, with data sampled using Pescador~1.1~\cite{brian_mcfee_2017_848831}.
Models were trained on mini-batches of $16$ 10-second patches.
Early stopping was used if the validation accuracy did not improve for 30 epochs; learning rate reduction was performed if the validation accuracy did not improve for 10 epochs.
Auto-pool models (including CAP and RAP) were initialized with $\alpha=1$.

All models were evaluated using the sed\_eval package~\cite{Mesaros_2016} to compute segment-based metrics with the segment duration set to 1~s as per the DCASE~2017 challenge evaluation.
For comparison purposes, we report accuracy for static (bag-level) prediction accuracy using the decision rule given in \cref{eq:bagpredict}, \rev{\emph{i.e.}, the maximum over dynamic predictions}.

When training on the MedleyDB dataset, training samples were generated by randomly sampling 10~second excerpts from the full-duration songs.
The bag label for each excerpt was considered positive for any instruments which were active for at least 10\% (1~s) of the excerpt, to match the 1~s duration used for the segment-based evaluation.

For reproducibility, we make our implementation and experiment framework software used in this study publicly available.\footnote{\url{https://github.com/marl/milsed}} 
To enable easy use of the proposed auto-pool function in new work, we have also implemented it as an independent Keras layer.\footnote{\url{https://github.com/marl/autopool}}

\section{Discussion}

This section describes the results of the experimental evaluation, broken down by data-set.

\subsection{URBAN-SED results}
\begin{table}
\centering
\caption{Class-aggregated results on URBAN-SED\label{tab:urbansed}.}
\begin{tabular}{lrrrrrrr}
\toprule
& \multicolumn{3}{c}{Static}
& \multicolumn{4}{c}{Dynamic}\\
Model 
		& $F_1$ & $P$ & $R$
		& $F_1$ & $P$ & $R$ & $E_\downarrow$\\
\cmidrule( r ){1-1}
\cmidrule( r ){2-4}
\cmidrule( r ){5-8}
Max 					& 0.742 & 0.774 & 0.717 & 0.463 & 0.774 & 0.330 & 0.695 \\
Mean 					& 0.543 & 0.726 & 0.436 & 0.408 & 0.280 & 0.751 & 2.10  \\
Soft-max				& 0.630 & 0.772 & 0.537 & 0.492 & 0.397 & 0.646 & 1.22  \\
\midrule
RAP $10^{-2}$	        & 0.544 & 0.719 & 0.449 & 0.419 & 0.296 & 0.717 & 1.88 \\
RAP $10^{-3}$	        & 0.746 & 0.790 & 0.711 & 0.529 & 0.584 & 0.484 & 0.731 \\
RAP $10^{-4}$	        & 0.754 & 0.754 & 0.756 & 0.526 & 0.650 & 0.442 & 0.681 \\
CAP 					& 0.754 & 0.781 & 0.732 & 0.533 & 0.622 & 0.466 & 0.696 \\
Auto 					& 0.757 & 0.784 & 0.739 & 0.504 & 0.738 & 0.382 & 0.665 \\
\midrule
Strong 					& 0.762 & 0.708 & 0.822 & 0.551 & 0.693 & 0.458 & 0.642 \\
\bottomrule
\end{tabular}
\end{table}

\Cref{tab:urbansed} presents the results of the URBAN-SED evaluation, \rev{averaged} across all classes.
On the static prediction task, auto-pool achieves the highest $F_1$ score of all MIL models under comparison, although the constrained and regularized variants are nearly equivalent.
Note that the \emph{strong} model, trained with full access to time-varying labels, performs only slightly better, indicating that the auto-pool is effective for static prediction.

This trend carries over to the dynamic prediction task, where the constrained auto-pool model (CAP) achieves $F_1=0.533$, compared to the strong model's $F_1=0.551$,
\rev{and comparable scores are achieved by the regularized models with $\lambda \in \{10^{-3}, 10^{-4}\}$.
On this dataset, the auto-pool model appears to over-fit} the weak annotations, and a similar trend can be observed for the $\max$-pooling model.
Conversely, RAP with $\lambda=10^{-2}$ appears to be over-regularized, and behaves similarly to mean-pooling on both static and dynamic prediction tasks.

\begin{figure*}
\centering
\includegraphics[width=0.8\textwidth]{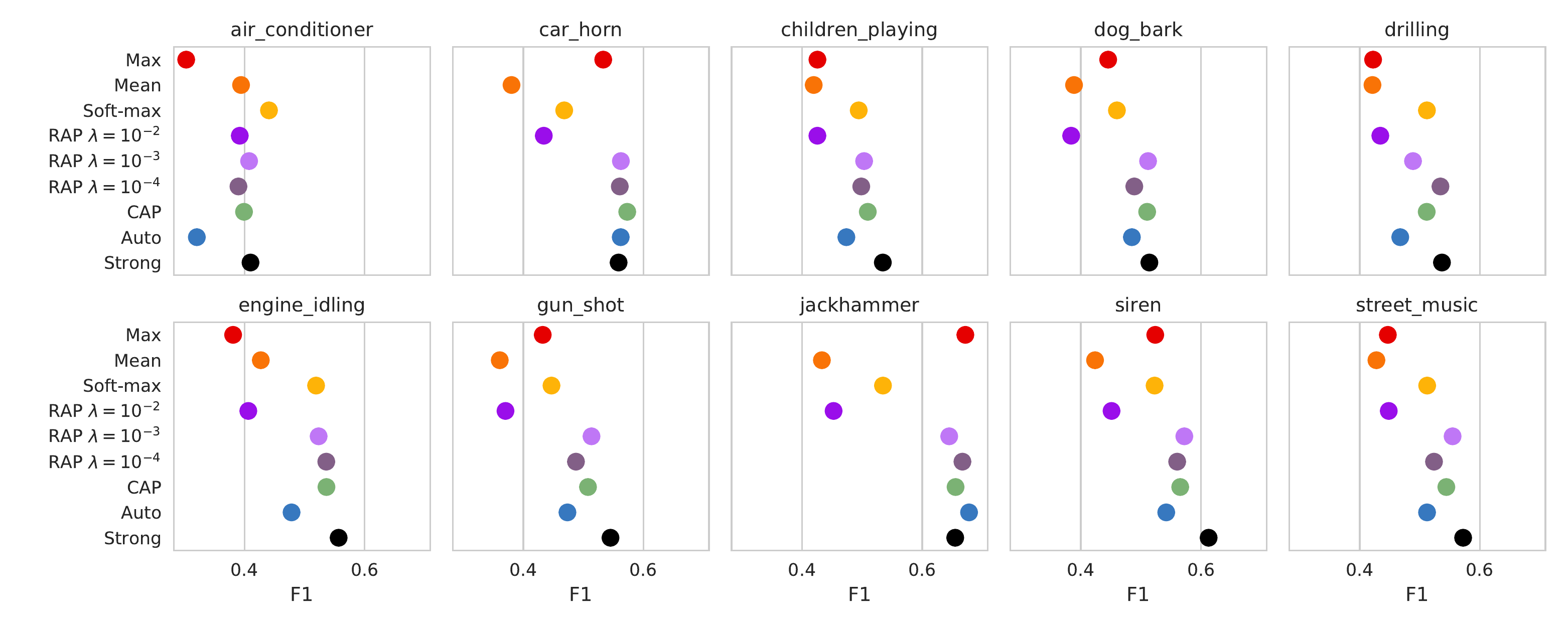}
\caption{URBAN-SED  results: per-class dynamic $F_1$ scores for each model under comparison.\label{fig:urbansedclass}}
\end{figure*}

\Cref{fig:urbansedclass} shows the $F_1$ scores independently for each class.
While there is some variation across classes, RAP ($\lambda\leq10^{-3}$) and CAP \rev{consistently achieve high scores}, and closely track the strong model.
Mean and RAP ($\lambda=10^{-2}$) tend to do poorly on event classes which are transient or highly localized in time (\emph{gun shot}, \emph{car horn}).
This is in accordance with \Cref{fig:milgrads}: mean-pooling predictions of sparse event categories assigns equal responsibility to each frame in the input, which will be erroneous for any frames that do not cover the event in question.
The fact that RAP $\lambda=10^{-2}$ exhibits this behavior indicates that the regularization term is too strong, and the model reverts to mean pooling.

\begin{figure}
\centering
\includegraphics[width=0.8\columnwidth]{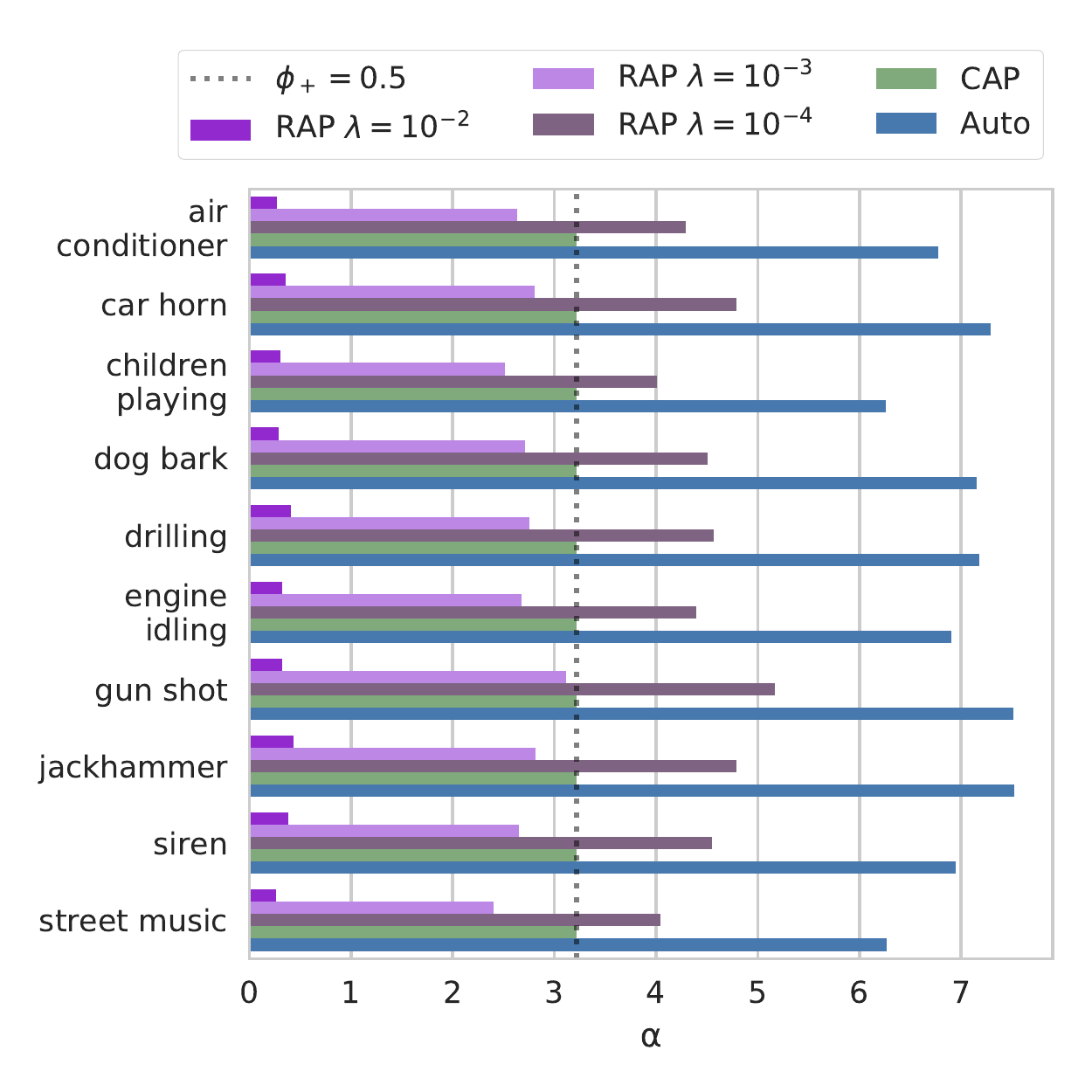}
\caption{URBAN-SED results: learned $\alpha$ parameters for each event class, for auto-pool, constrained auto-pool (CAP), and regularized auto-pool (RAP).
\label{fig:urbansedalpha}}
\end{figure}

\Cref{fig:urbansedalpha} illustrates the $\alpha$ vectors learned by each auto-pooling model.
In particular, the CAP model learns to maximize all $\alpha$ to the upper bound, indicating that max-like behavior is preferred for all classes.
This is likely an artifact of how the dataset was constructed: events are artificially clipped to at most 3~seconds, which results in implicitly sparse class activations for each bag (\Cref{fig:urbansed:durations}).
Note, however, that although the auto-pool models learn to produce max-like behavior, they consistently outperform the $\max$-pool model on this dataset.
This finding is consistent with the motivations for soft-max pooling given in \Cref{sec:methods}: max-pooling produces extremely sparse gradients during training, which impedes the model's ability to learn stable representations.
By contrast, initializing the auto-pool model with $\alpha=1$ (softmax-like behavior) produces dense gradients early in training, which become sparser as the model converges toward max-like behavior.

\begin{figure}
\centering%
\includegraphics[width=1.1\columnwidth]{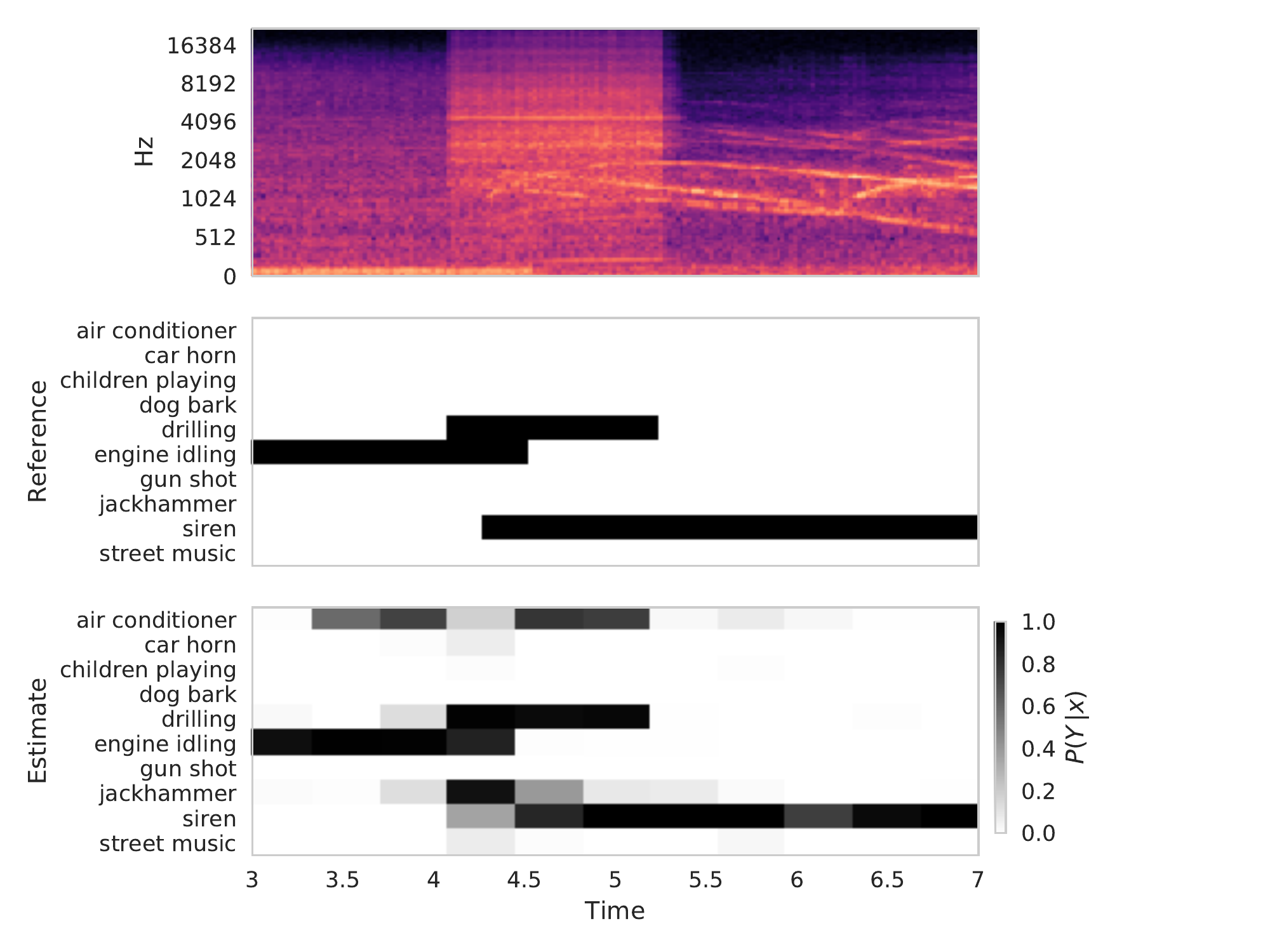}
\caption{Dynamic predictions made by the RAP model $\left(\lambda = 10^{-3}\right)$ on a validation clip from URBAN-SED.\@
Top: the input mel spectrogram; middle: the (dynamic) reference annotations; bottom: the predicted label likelihoods.\label{fig:urbansedpred}}
\end{figure}

\Cref{fig:urbansedpred} illustrates the predictions made by the RAP model with $\lambda={10}^{-3}$ on a validation clip.
While the model does show some confusion (\emph{engine\_idling} and \emph{air\_conditioner}, or \emph{drilling} and \emph{jackhammer}), the temporal localization is generally good.

\subsection{DCASE~2017 results}
\begin{table}
\centering
\caption{Aggregate results on DCASE~2017\label{tab:dcase}.}
\begin{tabular}{lrrrrrrr}
\toprule
& \multicolumn{3}{c}{Static}
& \multicolumn{4}{c}{Dynamic}\\
Model 
		& $F_1$ & $P$ & $R$
		& $F_1$ & $P$ & $R$ & $E_\downarrow$\\
\cmidrule( r ){1-1}
\cmidrule( r ){2-4}
\cmidrule( r ){5-8}
Max 					& 0.257 & 0.650 & 0.267 & 0.252 & 0.679 & 0.155 & 0.874 \\
Mean 					& 0.397 & 0.712 & 0.384 & 0.426 & 0.309 & 0.685 & 1.57 \\
Soft-max 				& 0.389 & 0.683 & 0.381 & 0.466 & 0.391 & 0.576 & 1.04 \\
\midrule
RAP $10^{-2}$	& 0.355 & 0.696 & 0.359 & 0.436 & 0.325 & 0.663 & 1.43 \\
RAP $10^{-3}$	& 0.357 & 0.669 & 0.357 & 0.410 & 0.308 & 0.613 & 1.44 \\
RAP $10^{-4}$	& 0.372 & 0.694 & 0.374 & 0.445 & 0.340 & 0.642 & 1.32 \\
CAP 					& 0.426 & 0.700 & 0.414 & 0.427 & 0.360 & 0.524 & 1.12 \\
Auto 					& 0.454 & 0.664 & 0.453 & 0.425 & 0.401 & 0.451 & 0.968 \\
\bottomrule
\end{tabular}
\end{table}

\Cref{tab:dcase} presents the class-aggregated results on the DCASE~2017 data.
Note that because the DCASE training data only has clip-level annotations, we cannot compare to a baseline model trained on strong annotations.
As before on URBAN-SED, the auto-pool method achieves the highest static $F_1$ score.
Soft-max pooling achieves the highest dynamic $F_1$ score (0.466), but both the mean and auto-pool methods are comparable, all landing in the range of 0.41--0.45.

\begin{figure*}
\centering
\includegraphics[width=\textwidth]{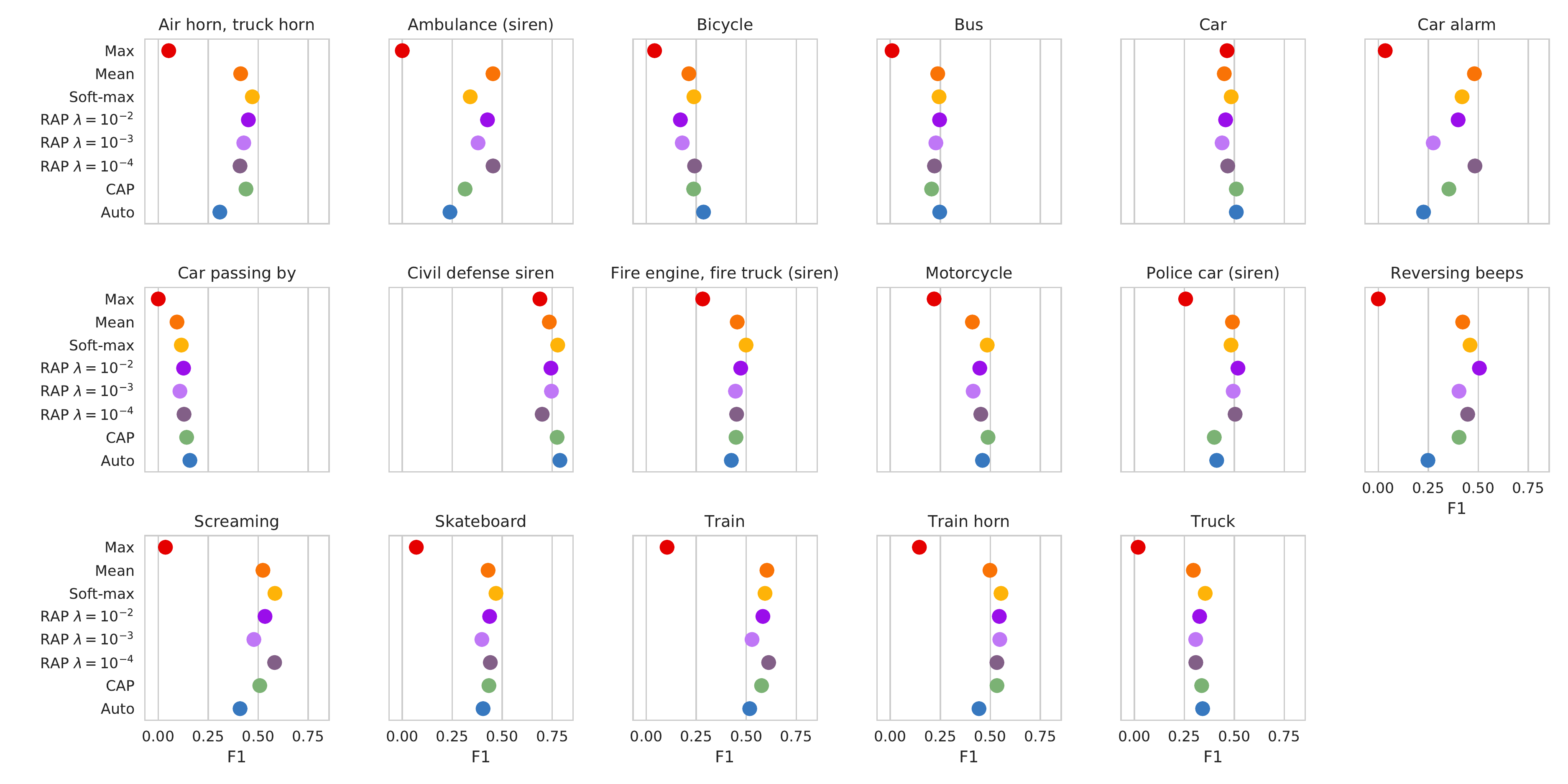}
\caption{DCASE~2017 results: per-class dynamic $F_1$ scores for each model under comparison.\label{fig:dcaseclass}}
\end{figure*}

Notably, the max-pooling model substantially under-performs the competing methods on both static and dynamic prediction tasks.
This holds uniformly across all per-class evaluations, as illustrated in \Cref{fig:dcaseclass}.
With the exception of unconstrained auto-pool, the remaining models generally perform comparably across all classes.

\begin{figure}
\centering%
\includegraphics[width=0.8\columnwidth]{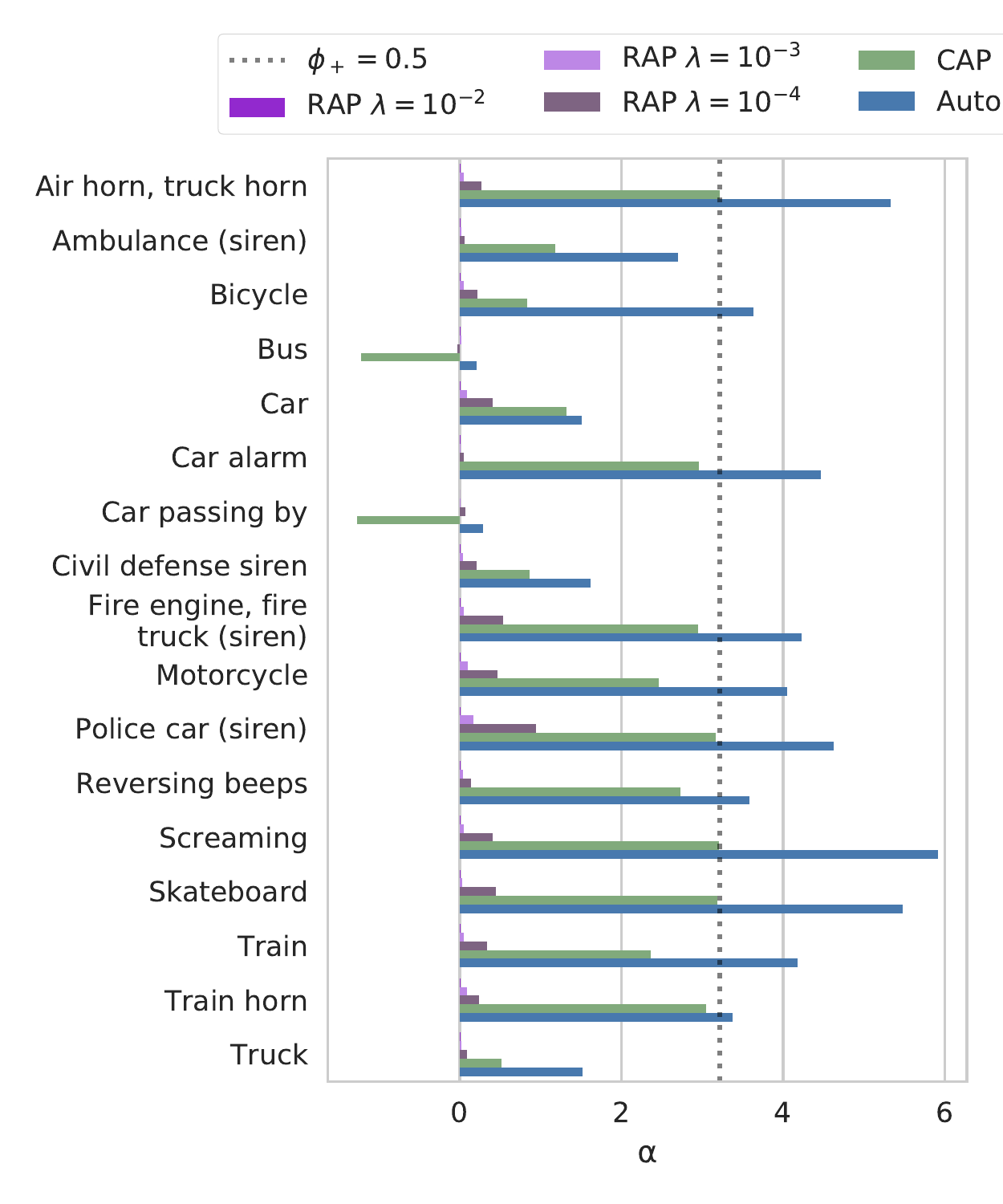}%
\caption{DCASE~2017 results: learned $\alpha$ parameters for each event class, for auto-pool, constrained auto-pool (CAP), and regularized auto-pool (RAP).\@
\label{fig:dcasealpha}}
\end{figure}

\Cref{fig:dcasealpha} shows the learned $\alpha$ vectors for each auto-pool model.
Unlike the URBAN-SED results in \Cref{fig:urbansedalpha}, \rev{auto-pool} models do not uniformly approach max-pooling \rev{on the DCASE data}.
Instead, there is significant diversity among the different classes, with some tending toward max-pooling behavior (large $\alpha$ for \emph{screaming} or \emph{air horn/truck horn}, \emph{skateboard}) and others tending toward mean-pooling behavior (small $\alpha$ for \emph{bus} or \emph{car passing by}, \emph{truck}).
Referring to \Cref{fig:dcase:durations}, the classes for which auto-pool (and CAP) learn large $\alpha$ tend to have short event durations.
By contrast, the classes which result in small $\alpha$ values tend to span the majority of the clip, and have high concentration on full duration (1.0).
In these classes, the bag- and instance-labels are equivalent, so it is expected that mean-pooling (small $\alpha$) out-performs $\max$-pooling.

\subsection{MedleyDB results}
\begin{table}
\centering
\caption{Aggregate results on MedleyDB over 10 randomized trials\label{tab:medley}. Results which are statistically indistinguishable from the best (average) per metric (underlined) are indicated in bold.}
\begin{tabular}{lrrrrrrr}
\toprule
& \multicolumn{3}{c}{Static}
& \multicolumn{4}{c}{Dynamic}\\
Model 
		& $F_1$ & $P$ & $R$
		& $F_1$ & $P$ & $R$ & $E_\downarrow$\\
\cmidrule( r ){1-1}
\cmidrule( r ){2-4}
\cmidrule( r ){5-8}
Max           & \textbf{0.650}            & \underline{\textbf{0.605}} & 0.829                      & 0.437                      & \underline{\textbf{0.875}} & 0.292                      & 0.719 \\
Mean          & 0.550                      & 0.409                      & \textbf{0.988}             & \textbf{0.655}             & 0.594                      & \underline{\textbf{0.733}} & \textbf{0.608} \\
Soft-max      & \textbf{0.577}             & 0.444                      & \textbf{0.974}             & \textbf{0.662}             & 0.668                      & \textbf{0.658}             & \underline{\textbf{0.524}} \\
\midrule
RAP $10^{-2}$ & 0.553                      & 0.413                      & \underline{\textbf{0.989}} & \textbf{0.659}             & 0.604                      & \textbf{0.727}            & \textbf{0.593}\\
RAP $10^{-3}$ & 0.563                      & 0.425                      & \textbf{0.984}             & \underline{\textbf{0.673}} & 0.638                      & \textbf{0.714}             & \textbf{0.545} \\
RAP $10^{-4}$ & \textbf{0.623}            & \textbf{0.497}            & \textbf{0.957}             & \textbf{0.622}             & \textbf{0.757}             & 0.530                      & \textbf{0.540} \\
CAP           & \textbf{0.625}             & \textbf{0.512}             & \textbf{0.937}             & \textbf{0.609}             & \textbf{0.787}             & 0.498                      & \textbf{0.551} \\
Auto          & \underline{\textbf{0.653}} & \textbf{0.567}            & 0.888                      & 0.528                      & \textbf{0.841}             & 0.386                      & 0.636 \\
\midrule
Strong        & 0.575                      & 0.437                      & 0.982                      & 0.675                      & 0.640                      & 0.716                      & 0.540 \\
\bottomrule
\end{tabular}
\end{table}

\Cref{tab:medley} lists the class-aggregated scores over the MedleyDB dataset.
Following Dem\v{s}ar~\cite{demvsar2006statistical}, the distributions of scores over all splits are compared using a Friedman test~\cite{friedman1937use} with Bonferroni-Holm correction ($\alpha=0.05$)~\cite{holm1979simple}, and methods indistinguishable from the best (average) are indicated in bold.
The strong model is omitted from statistical comparison, as we are primarily concerned with differentiating among MIL algorithms.
From this analysis, we observe little differentiation between the various methods.
Mean-pooling and RAP ($\lambda\geq 10^{-3}$) are significantly worse than auto-pool (best) for static $F_1$ score, though still comparable to the strong model.
For dynamic prediction, only the max- and auto-pooling methods are significantly worse than RAP $\lambda=10^{-3}$, which closely matches the strong model.

\begin{figure*}
\centering
\includegraphics[width=0.8\textwidth]{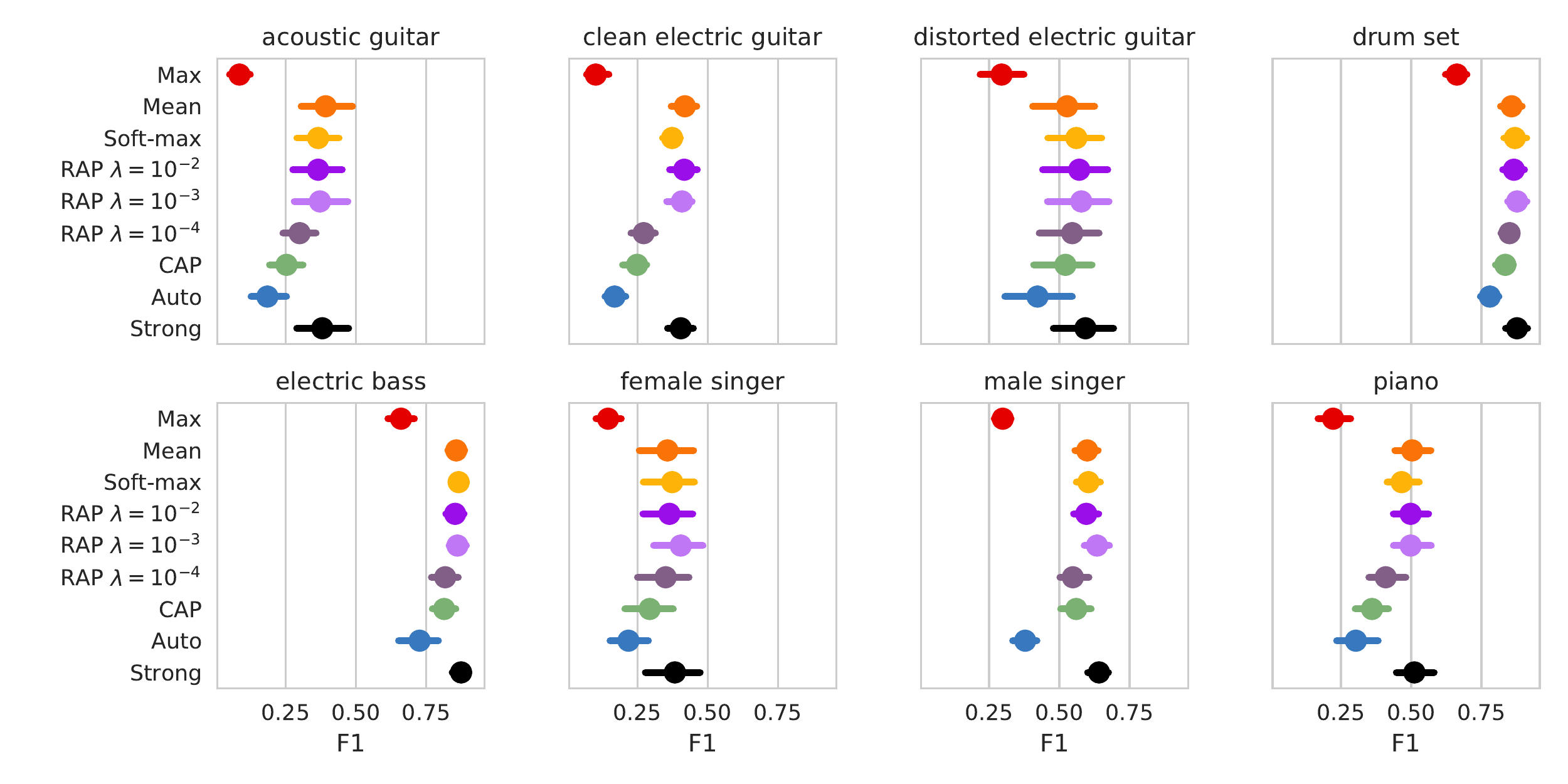}
\caption{MedleyDB results: per-class dynamic $F_1$ scores for each model under comparison, averaged over 10 randomized trials. Error bars correspond to 95\% confidence intervals under bootstrap sampling.\label{fig:medleyclass}}
\end{figure*}

An examination of the per-class results presented in \Cref{fig:medleyclass} reveals that this trend is consistent across classes.
The low performance of max-pooling exhibited on the DCASE dataset persists on MedleyDB.\@
Similarly, the auto-pool model tends to do worse than the regularized variants across all classes.
This is most likely due to the characteristics of the training data: instruments within a randomly selected excerpt tend to be either entirely active or inactive, so mean-pooling is a good approximation to strong training.
This phenomenon is illustrated in \Cref{fig:medley:durations}, which shows the distribution of labeled segment durations for each instrument.
Aside from vocalists, the average durations are well in excess of the 10~second mark (gray), which indicates that uniformly sampled patches are unlikely to catch instrument state transitions.

\section{Conclusion}

To summarize the experimental results presented above, we observe the following trends across all datasets.
First, the unconstrained, unregularized auto-pool method consistently achieves the highest scores for static prediction.
If the practitioner's goal is to classify weakly labeled excerpts without requiring more precise prediction, then auto-pool appears to be the method of choice.
However, auto-pool does exhibit a tendency to ``over-fit'' to weak annotation, in that its performance for dynamic prediction is generally \rev{lower} than the proposed alternatives, and that it favors precision over recall.

Second, the behavior of fixed pooling operators (min, max, soft-max) depends on the characteristics of the dataset and the relative duration of events in each class.
Mean-pooling performs well when events are long relative to the bag because the bag-level labels can reasonably be propagated to all instances.
Max-pooling can perform well when events are short within the bag, but it can also be unstable and difficult to train.
While auto-pooling often converges to max-like behavior, it consistently outperforms the standard max-pool model, which indicates that the improved gradient flow due to the soft-max operator is indeed beneficial for learning good representations.

\rev{%
Third, as a general observation, $\max$-pooling models tend to favor precision over recall in dynamic evaluation.
This is likely due to the fact that to optimize the objective during training, $\max$-pooling needs only to model a single instance within a bag.
This obviously suffices for static evaluation, but for dynamic evaluation, $\max$-pooling models have no incentive to model the entire duration of the source event, leading to a reduction of recall.
Similarly, the more $\max$-like the pooling operator becomes, \emph{e.g.}, RAP with small $\lambda$ or unconstrained auto-pool, the more emphasis the resulting model tends to place on precision rather than recall.
For similar reasons, strongly trained models can under-perform MIL models in static evaluation, as illustrated in \cref{tab:medley}.
MIL models can attend to specific portions of non-stationary signals (\emph{e.g.}, a vocal attack) to detect their presence, while strongly trained models attempt to solve the more difficult task of modeling the entire duration of the event.
}

\rev{%
Although not empirically studied in this work, the choice of initialization for $\alpha$ could also influence the resulting model.
Following the motivation given in \Cref{sec:methods}, we generally recommend to initialize $\alpha$ with small values (either 0 or 1)
to ensure sufficient gradient propagation early in training.}

In all datasets, the regularized auto-pool models are among the best performing, illustrating that the models are able to adapt to the characteristics of the data for a proper choice of $\lambda$.
This suggests a general recommendation for MIL event detection problems: use RAP, and tune $\lambda$ by hyper-parameter optimization over a strongly labeled validation set.

Most importantly, the proposed method is able to nearly match dynamic prediction accuracy to that obtained by training with access to instance labels.
This suggests that by framing sound event detection as a MIL problem, practitioners may be able to achieve comparable accuracy with a significant reduction in effort and cost of acquiring training labels.
\rev{Finally, although we focus on SED applications in this article, we emphasize that the proposed auto-pool operators are fully general, and could be readily applied to MIL problems in any application domain.}


%



\section*{Acknowledgment}
The authors acknowledge support from the Moore-Sloan Data Science Environment at NYU.\@
This work was partially supported by NSF awards 1544753 and 1633259, and a Google Faculty Award.
We thank NVidia Corporation for the donation of a Tesla K40 GPU.\@

\ifCLASSOPTIONcaptionsoff%
  \newpage
\fi



\bibliographystyle{IEEEtran}
\small
\bibliography{refs}

\begin{thebibliography}{10}
\providecommand{\url}[1]{#1}
\csname url@samestyle\endcsname
\providecommand{\newblock}{\relax}
\providecommand{\bibinfo}[2]{#2}
\providecommand{\BIBentrySTDinterwordspacing}{\spaceskip=0pt\relax}
\providecommand{\BIBentryALTinterwordstretchfactor}{4}
\providecommand{\BIBentryALTinterwordspacing}{\spaceskip=\fontdimen2\font plus
\BIBentryALTinterwordstretchfactor\fontdimen3\font minus
  \fontdimen4\font\relax}
\providecommand{\BIBforeignlanguage}[2]{{%
\expandafter\ifx\csname l@#1\endcsname\relax
\typeout{** WARNING: IEEEtran.bst: No hyphenation pattern has been}%
\typeout{** loaded for the language `#1'. Using the pattern for}%
\typeout{** the default language instead.}%
\else
\language=\csname l@#1\endcsname
\fi
#2}}
\providecommand{\BIBdecl}{\relax}
\BIBdecl

\bibitem{Bello:SONYC:CACM:18}
J.~P. Bello, C.~Silva, O.~Nov, R.~L. DuBois, A.~Arora, J.~Salamon, C.~Mydlarz,
  and H.~Doraiswamy, ``{SONYC}: A system for the monitoring, analysis and
  mitigation of urban noise pollution,'' \emph{Communications of the ACM}, In
  press, 2018.

\bibitem{Stowell:AEDoverlappinh:WASPAA:15}
D.~Stowell, , and D.~Clayton, ``Acoustic event detection for multiple
  overlapping similar sources,'' in \emph{IEEE WASPAA}, Oct. 2015, pp. 1--5.

\bibitem{Salamon:FlightCalls:PLOSONE:16}
J.~Salamon, J.~P. Bello, A.~Farnsworth, M.~Robbins, S.~Keen, H.~Klinck, and
  S.~Kelling, ``Towards the automatic classification of avian flight calls for
  bioacoustic monitoring,'' \emph{PLOS ONE}, vol.~11, no.~11, p. e0166866,
  2016.

\bibitem{Lostanlen:BirdVoxFullNight:ICASSP:18}
V.~Lostanlen, J.~Salamon, A.~Farnsworth, S.~Kelling, and J.~P. Bello,
  ``Birdvox-full-night: A dataset and benchmark for avian flight call
  detection,'' in \emph{IEEE Int.~Conf.~on Acoustics, Speech and Signal
  Processing (ICASSP)}, Apr. 2018.

\bibitem{DCASE17:SmartCars:Webpage:17}
\BIBentryALTinterwordspacing
Large-scale weakly supervised sound event detection for smart cars. [Online].
  Available:
  \url{http://www.cs.tut.fi/sgn/arg/dcase2017/challenge/task-large-scale-sound-event-detection}
\BIBentrySTDinterwordspacing

\bibitem{radhakrishnan:AudioSurveillance:WASPAA:05}
R.~Radhakrishnan, A.~Divakaran, and P.~Smaragdis, in \emph{IEEE WASPAA}.

\bibitem{Crocco:AudioSuerveillance:CS:16}
M.~Crocco, M.~Cristani, A.~Trucco, and V.~Murino, ``Audio surveillance: A
  systematic review,'' \emph{ACM Comput.~Surv.}, vol.~48, no.~4, pp.
  52:1--52:46, 2016.

\bibitem{Goetze:HealthSED:JCSE:12}
S.~Goetze, J.~Schroder, S.~Gerlach, D.~Hollosi, J.-E. Appell, and F.~Wallhoff,
  ``Acoustic monitoring and localization for social care,'' \emph{Journal of
  Computing Science and Engineering}, vol.~6, no.~1, pp. 40--50, 2012.

\bibitem{Hershey:LargeAudioCNN:ICASSP:17}
S.~Hershey, S.~Chaudhuri, D.~P.~W. Ellis, J.~Gemmeke, A.~Jansen, C.~Moore,
  M.~Plakal, D.~Platt, R.~Saurous, B.~Seybold, M.~Slaney, R.~Weiss, and
  K.~Wilson, ``{CNN} architectures for large-scale audio classification,'' in
  \emph{IEEE Int.~Conf.~on Acoustics, Speech and Signal Processing (ICASSP)},
  Mar. 2017, pp. 131--135.

\bibitem{mesaros2017dcase}
A.~Mesaros, T.~Heittola, A.~Diment, B.~Elizalde, A.~Shah, E.~Vincent, B.~Raj,
  and T.~Virtanen, ``Dcase 2017 challenge setup: Tasks, datasets and baseline
  system,'' in \emph{DCASE 2017-Workshop on Detection and Classification of
  Acoustic Scenes and Events}, 2017.

\bibitem{Stowell:BioacousticCASA:Chapter:18}
D.~Stowell, ``Computational bioacoustic scene analysis,'' in
  \emph{Computational Analysis of Sound Scenes and Events}, T.~Virtanen, M.~D.
  Plumbley, and D.~Ellis, Eds.\hskip 1em plus 0.5em minus 0.4em\relax Springer
  International Publishing, 2018, pp. 303--333.

\bibitem{Temko:AED:PhD:07}
A.~Temko, ``Acoustic event detection and classification,'' Ph.D. dissertation,
  Department of Signal Theory and Communications, Universitat Politecnica de
  Catalunya, Barcelona, Spain, 2007.

\bibitem{Foggia:SED:PRL:15}
P.~Foggia, N.~Petkov, A.~Saggese, N.~Strisciuglio, and M.~Vento, ``Reliable
  detection of audio events in highly noisy environments,'' \emph{Pattern
  Recognition Letters}, vol.~65, pp. 22--28, 2015.

\bibitem{Elizalde:SED:DCASE:16}
B.~Elizalde, A.~Kumar, A.~Shah, R.~Badlani, E.~Vincent, B.~Raj, and I.~Lane,
  ``Experiments on the {DCASE} challenge 2016: Acoustic scene classification
  and sound event detection in real life recording,'' in \emph{Proc.~DCASE
  Workshop}, Budapest, Hungary, Sep. 2016, pp. 20--24.

\bibitem{cai:KeyAudioEffects:TASLP:06}
L.-H. Cai, L.~Lu, A.~Hanjalic, H.-J. Zhang, and L.-H. Cai, ``A flexible
  framework for key audio effects detection and auditory context inference,''
  \emph{IEEE Trans.~on Audio, Speech, and Language Processing}, vol.~14, no.~3,
  pp. 1026--1039, May 2006.

\bibitem{Mesaros:AED:EUSIPCO:10}
A.~Mesaros, T.~Heittola, A.~Eronen, and T.~Virtanen, ``Acoustic event detection
  in real life recordings,'' in \emph{European Signal Processing Conference
  (EUSIPCO)}, Aalborg, Denmark, 2010.

\bibitem{heittola:ContextEventDetection:EURASIP:13}
T.~Heittola, A.~Mesaros, A.~Eronen, and T.~Virtanen, ``Context-dependent sound
  event detection,'' \emph{EURASIP J.~on Audio, Speech and Music Processing},
  vol. 2013, no.~1, 2013.

\bibitem{Vuegen:AEDGMM:WASPAA:13}
L.~Vuegen, B.~V.~D. Broeck, P.~Karsmakers, J.~F. Gemmeke, B.~Vanrumste, and
  H.~V. Hamme, ``An {MFCC}-{GMM} approach for event detection and
  classification,'' in \emph{IEEE WASPAA}, 2013.

\bibitem{Benetos:AED:ICASSP:16}
E.~Benetos, G.~Lafay, M.~Lagrange, and M.~D. Plumbley, ``Detection of
  overlapping acoustic events using a temporally-constrained probabilistic
  model,'' in \emph{IEEE International Conference on Acoustics, Speech and
  Signal Processing (ICASSP)}, Shanghai, China, 2016.

\bibitem{Benetos:SEDPLCA:TASLP:17}
------, ``Polyphonic sound event tracking using linear dynamical systems,''
  \emph{IEEE/ACM Transactions on Audio, Speech, and Language Processing},
  vol.~25, no.~6, pp. 1266--1277, Jun. 2017.

\bibitem{heittola2013supervised}
T.~Heittola, A.~Mesaros, T.~Virtanen, and M.~Gabbouj, ``Supervised model
  training for overlapping sound events based on unsupervised source
  separation,'' in \emph{IEEE Int.~Conf.~on Acoustics, Speech and Signal
  Processing (ICASSP)}, May 2013.

\bibitem{cotton:SpectroTemporal:WASPAA:11}
C.~V. Cotton and D.~P.~W. Ellis, ``Spectral vs. spectro-temporal features for
  acoustic event detection,'' in \emph{IEEE WASPAA}, Oct. 2011, pp. 69--72.

\bibitem{Dikmen:SEDNMF:WASPAA:13}
O.~Dikmen and A.~Mesaros, ``Sound event detection using non-negative
  dictionaries learned from annotated overlapping events,'' in \emph{IEEE
  WASPAA}, 2013.

\bibitem{Gemmeke:AEDNMF:WASPAA:13}
J.~F. Gemmeke, L.~Vuegen, P.~Karsmakers, B.~Vanrumste, and H.~V. hamme, ``An
  exemplar-based nmf approach to audio event detection,'' in \emph{IEEE
  WASPAA}, Oct. 2013.

\bibitem{Mesaros:SEDNMF:ICASSP:15}
A.~Mesaros, T.~Heittola, O.~Dikmen, and T.~Virtanen, ``Sound event detection in
  real life recordings using coupled matrix factorization of spectral
  representations and class activity annotations,'' in \emph{International
  Conference on Acoustics, Speech and Signal Processing (ICASSP)}, Brisbane,
  Australia, 2015.

\bibitem{Komatsu:AEDNMF:ICASSP:16}
T.~Komatsu, Y.~Senda, and R.~Kondo, ``Acoustic event detection based on
  non-negative matrix factorization with mixtures of local dictionaries and
  activation aggregation,'' in \emph{IEEE International Conference on
  Acoustics, Speech and Signal Processing (ICASSP)}, Shanghai, China, Mar.
  2016, pp. 2259--2263.

\bibitem{Cakir:SEDDNN:IJCNN:15}
E.~Cakir, T.~Heittola, H.~Huttunen, and T.~Virtanen, ``Polyphonic sound event
  detection using multi label deep neural networks,'' in \emph{2015
  International Joint Conference on Neural Networks (IJCNN)}, July 2015, pp.
  1--7.

\bibitem{Cakir:FilterbankSED:IJCNN:16}
E.~Cakir, E.~C. Ozan, and T.~Virtanen, ``Filterbank learning for deep neural
  network based polyphonic sound event detection,'' in \emph{International
  Joint Conference on Neural Networks (IJCNN)}, Jul. 2016, pp. 3399--3406.

\bibitem{Jeong:AEDCNN:DCASE:17}
I.-Y. Jeong, S.~Lee, Y.~Han, and K.~Lee, ``Audio event detection using
  multiple-input convolutional neural network,'' in \emph{Proceedings of the
  Detection and Classification of Acoustic Scenes and Events 2017 Workshop
  (DCASE2017)}, 2017.

\bibitem{Parascandolo:RNNSED:ICASSP:16}
G.~Parascandolo, H.~Huttunen, and T.~Virtanen, ``Recurrent neural networks for
  polyphonic sound event detection in real life recordings,'' in
  \emph{International Conference on Acoustics, Speech and Signal Processing
  (ICASSP)}, Shanghai, China, Mar. 2016, pp. 6440--6444.

\bibitem{Rui:SEDRNN:DCASE:17}
R.~Lu and Z.~Duan, ``Bidirectional {GRU} for sound event detection,'' DCASE
  2017 challenge, extended abstract, Tech. Rep., 2017.

\bibitem{Cakir:SEDCRNN:TASLP:17}
E.~\c{C}akir, G.~Parascandolo, T.~Heittola, H.~Huttunen, and T.~Virtanen,
  ``Convolutional recurrent neural networks for polyphonic sound event
  detection,'' \emph{IEEE/ACM Trans.~on Audio, Speech and Lang.~Proc., Special
  Issue on Sound Scene and Event Analysis}, In press, 2017.

\bibitem{Adavanne:SpatialSEDCRNN:ICASSP:17}
S.~Adavanne, P.~Pertil{\"a}, and T.~Virtanen, ``Sound event detection using
  spatial features and convolutional recurrent neural network,'' in \emph{IEEE
  International Conference on Acoustics, Speech and Signal Processing
  (ICASSP)}, 2017.

\bibitem{salamon2017scaper}
J.~Salamon, D.~MacConnell, M.~Cartwright, P.~Li, and J.~P. Bello, ``Scaper: A
  library for soundscape synthesis and augmentation,'' in \emph{IEEE WASPAA},
  Oct. 2017, pp. 344--348.

\bibitem{Gemmeke:AudioSet:ICASSP:17}
J.~F. Gemmeke, D.~P.~W. Ellis, D.~Freedman, A.~Jansen, W.~Lawrence, R.~C.
  Moore, M.~Plakal, and M.~Ritter, ``Audio set: An ontology and human-labeled
  dataset for audio events,'' in \emph{IEEE Int.~Conf.~on Acoustics, Speech and
  Signal Processing (ICASSP)}, Mar. 2017, pp. 776--780.

\bibitem{dietterich1997solving}
T.~G. Dietterich, R.~H. Lathrop, and T.~Lozano-P{\'e}rez, ``Solving the
  multiple instance problem with axis-parallel rectangles,'' \emph{Artificial
  intelligence}, vol.~89, no. 1-2, pp. 31--71, 1997.

\bibitem{zhang2006multiple}
C.~Zhang, J.~C. Platt, and P.~A. Viola, ``Multiple instance boosting for object
  detection,'' in \emph{Advances in neural information processing systems},
  2006, pp. 1417--1424.

\bibitem{babenko2011multiple}
B.~Babenko, N.~Verma, P.~Doll{\'a}r, and S.~J. Belongie, ``Multiple instance
  learning with manifold bags.'' in \emph{ICML}, 2011.

\bibitem{hsu2014augmented}
K.-J. Hsu, Y.-Y. Lin, and Y.-Y. Chuang, ``Augmented multiple instance
  regression for inferring object contours in bounding boxes,'' \emph{IEEE
  Transactions on Image Processing}, vol.~23, no.~4, pp. 1722--1736, 2014.

\bibitem{mandel2008multiple}
M.~I. Mandel and D.~P. Ellis, ``Multiple-instance learning for music
  information retrieval.'' in \emph{ISMIR}, 2008, pp. 577--582.

\bibitem{andrews2003support}
S.~Andrews, I.~Tsochantaridis, and T.~Hofmann, ``Support vector machines for
  multiple-instance learning,'' in \emph{Advances in neural information
  processing systems}, 2003, pp. 577--584.

\bibitem{chen2006miles}
Y.~Chen, J.~Bi, and J.~Z. Wang, ``Miles: Multiple-instance learning via
  embedded instance selection,'' \emph{IEEE Transactions on Pattern Analysis
  and Machine Intelligence}, vol.~28, no.~12, pp. 1931--1947, 2006.

\bibitem{wu2014music}
B.~Wu, E.~Zhong, A.~Horner, and Q.~Yang, ``Music emotion recognition by
  multi-label multi-layer multi-instance multi-view learning,'' in
  \emph{Proceedings of the 22nd ACM international conference on
  Multimedia}.\hskip 1em plus 0.5em minus 0.4em\relax ACM, 2014, pp. 117--126.

\bibitem{briggs2012acoustic}
F.~Briggs, B.~Lakshminarayanan, L.~Neal, X.~Z. Fern, R.~Raich, S.~J. Hadley,
  A.~S. Hadley, and M.~G. Betts, ``Acoustic classification of multiple
  simultaneous bird species: A multi-instance multi-label approach,'' \emph{The
  Journal of the Acoustical Society of America}, vol. 131, no.~6, pp.
  4640--4650, 2012.

\bibitem{zhang2005k}
M.-L. Zhang and Z.-H. Zhou, ``A k-nearest neighbor based algorithm for
  multi-label classification,'' in \emph{Granular Computing, 2005 IEEE
  International Conference on}, vol.~2.\hskip 1em plus 0.5em minus 0.4em\relax
  IEEE, 2005.

\bibitem{zhou2007multi}
Z.-H. Zhou and M.-L. Zhang, ``Multi-instance multi-label learning with
  application to scene classification,'' in \emph{Advances in neural
  information processing systems}, 2007, pp. 1609--1616.

\bibitem{Kumar:WeakSED:IJCNN:17}
A.~Kumar and B.~Raj, ``Audio event and scene recognition: A unified approach
  using strongly and weakly labeled data,'' in \emph{International Joint
  Conference on Neural Networks (IJCNN)}, May 2017, pp. 3475--3482.

\bibitem{Kumar:WeakAED:ACMMM:16}
------, ``Audio event detection using weakly labeled data,'' in
  \emph{Proceedings of the ACM Multimedia Conference (ACM-MM)}, Amsterdam, The
  Netherlands, Oct. 2016, pp. 1038--1047.

\bibitem{Kong:WeakSED:ICASSP:17}
Q.~Kong, Y.~Xu, W.~Wang, and M.~D. Plumbley, ``A joint detection-classification
  model for audio tagging of weakly labelled data,'' in \emph{IEEE
  International Conference on Acoustics, Speech and Signal Processing
  (ICASSP)}, Mar. 2017, pp. 641--645.

\bibitem{Su:WeakAED:ICASSP:17}
T.-W. Su, J.-Y. Liu, and Y.-H. Yang, ``Weakly-supervised audio event detection
  using event-specific gaussian filters and fully convolutional networks,'' in
  \emph{IEEE Int.~Conf.~on Acoustics, Speech and Signal Processing (ICASSP)},
  Mar. 2017, pp. 791--795.

\bibitem{Chou:FrameCNN:DCASE:17}
S.-Y. Chou, J.-S.~R. Jang, and Y.-H. Yang, ``Framecnn: a weakly-supervised
  learning framework for frame-wise acoustic event detection and
  classification,'' DCASE 2017 challenge, extended abstract, Tech. Rep., 2017.

\bibitem{Salamon:MILSED:DCASE:17}
J.~Salamon, B.~McFee, P.~Li, and J.~P. Bello, ``{DCASE} 2017 submission:
  Multiple instance learning for sound event detection,'' DCASE 2017 challenge,
  extended abstract, Tech. Rep., 2017.

\bibitem{Kumar:WeakSED:ICASSP:18}
A.~Kumar, M.~Khadkevich, and C.~Fugen, ``Knowledge transfer from weakly labeled
  audio using convolutional neural network for sound events and scenes,'' in
  \emph{IEEE Int.~Conf.~on Acoustics, Speech and Signal Processing (ICASSP)},
  Apr. 2018.

\bibitem{Xu:WeakSED:ICASSP:18}
Y.~Xu, Q.~Kong, W.~Wang, and M.~D. Plumbley, ``Large-scale weakly supervised
  audio classification using gated convolutional neural network,'' in
  \emph{IEEE Int.~Conf.~on Acoustics, Speech and Signal Processing (ICASSP)},
  Apr. 2018.

\bibitem{Wang:WeakSED:ICASSP:17}
Y.~Wang and F.~Metze, ``A first attempt at polyphonic sound event detection
  using connectionist temporal classification,'' in \emph{IEEE International
  Conference on Acoustics, Speech and Signal Processing (ICASSP)}, Mar. 2017,
  pp. 2986--2990.

\bibitem{Adavanne:WeakSED:DCASE:17}
S.~Adavanne and T.~Virtanen, ``Sound event detection using weakly labeled
  dataset with stacked convolutional and recurrent neural network,'' in
  \emph{Proceedings of the Detection and Classification of Acoustic Scenes and
  Events 2017 Workshop (DCASE2017)}, Nov. 2017.

\bibitem{Kong:WeakSED:ICASSP:18}
Q.~Kong, Y.~Xu, W.~Wang, and M.~D. Plumbley, ``A joint
  separation-classification model for sound event detection of weakly labelled
  data,'' in \emph{IEEE Int.~Conf.~on Acoustics, Speech and Signal Processing
  (ICASSP)}, Apr. 2018.

\bibitem{Sobieraj:WeakSED:ICASSP:18}
I.~Sobieraj, L.~Rencker, and M.~D. Plumbley, ``Orthogonality-regularized masked
  {NMF} for learning on weakly labeled audio data,'' in \emph{IEEE
  Int.~Conf.~on Acoustics, Speech and Signal Processing (ICASSP)}, Apr. 2018.

\bibitem{bahdanau2014neural}
D.~Bahdanau, K.~Cho, and Y.~Bengio, ``Neural machine translation by jointly
  learning to align and translate,'' in \emph{International Conference on
  Learning Representations}, ser. ICLR, 2015.

\bibitem{raffel2015feed}
C.~Raffel and D.~P. Ellis, ``Feed-forward networks with attention can solve
  some long-term memory problems,'' in \emph{International Conference on
  Learning Representations (workshop track)}, 2016.

\bibitem{zeiler2012differentiable}
M.~D. Zeiler and R.~Fergus, ``Differentiable pooling for hierarchical feature
  learning,'' \emph{arXiv preprint arXiv:1207.0151}, 2012.

\bibitem{swietojanski2015differentiable}
P.~Swietojanski and S.~Renals, ``Differentiable pooling for unsupervised
  speaker adaptation,'' in \emph{Acoustics, Speech and Signal Processing
  (ICASSP), 2015 IEEE International Conference on}.\hskip 1em plus 0.5em minus
  0.4em\relax IEEE, 2015, pp. 4305--4309.

\bibitem{Mesaros_2016}
\BIBentryALTinterwordspacing
A.~Mesaros, T.~Heittola, and T.~Virtanen, ``Metrics for polyphonic sound event
  detection,'' \emph{Applied Sciences}, vol.~6, no.~12, p. 162, May 2016.
  [Online]. Available: \url{http://dx.doi.org/10.3390/app6060162}
\BIBentrySTDinterwordspacing

\bibitem{Salamon:UrbanSound:ACMMM:14}
J.~Salamon, C.~Jacoby, and J.~P. Bello, ``A dataset and taxonomy for urban
  sound research,'' in \emph{22nd {ACM} International Conference on Multimedia
  (ACM-MM'14)}, Nov. 2014, pp. 1041--1044.

\bibitem{bittner2014medleydb}
R.~M. Bittner, J.~Salamon, M.~Tierney, M.~Mauch, C.~Cannam, and J.~P. Bello,
  ``Medleydb: A multitrack dataset for annotation-intensive mir research.'' in
  \emph{ISMIR}, vol.~14, 2014, pp. 155--160.

\bibitem{arandjelovic2017look}
R.~Arandjelovic and A.~Zisserman, ``Look, listen and learn,'' in \emph{2017
  IEEE International Conference on Computer Vision (ICCV)}.\hskip 1em plus
  0.5em minus 0.4em\relax IEEE, 2017, pp. 609--617.

\bibitem{Ioffe:BatchNorm:ICML:15}
S.~Ioffe and C.~Szegedy, ``Batch normalization: Accelerating deep network
  training by reducing internal covariate shift,'' in \emph{32nd Int.~Conf.~on
  Machine Learning}, ser. Proceedings of Machine Learning Research, F.~Bach and
  D.~Blei, Eds., vol.~37, Lille, France, Jul. 2015, pp. 448--456.

\bibitem{mcfee2015software}
B.~McFee, E.~J. Humphrey, and J.~P. Bello, ``A software framework for musical
  data augmentation.'' in \emph{ISMIR}.\hskip 1em plus 0.5em minus 0.4em\relax
  Citeseer, 2015, pp. 248--254.

\bibitem{mcfee_brian_2017_293021}
\BIBentryALTinterwordspacing
B.~McFee, M.~McVicar, O.~Nieto, S.~Balke, C.~Thome, D.~Liang, E.~Battenberg,
  J.~Moore, R.~Bittner, R.~Yamamoto, D.~Ellis, F.-R. Stoter, D.~Repetto,
  S.~Waloschek, C.~Carr, S.~Kranzler, K.~Choi, P.~Viktorin, J.~F. Santos,
  A.~Holovaty, W.~Pimenta, and H.~Lee, ``librosa 0.5.1,'' may 2017. [Online].
  Available: \url{http://doi.org/10.5281/zenodo.1022770}
\BIBentrySTDinterwordspacing

\bibitem{chollet2015keras}
F.~Chollet \emph{et~al.}, ``Keras,'' 2015.

\bibitem{abadi2016tensorflow}
M.~Abadi, P.~Barham, J.~Chen, Z.~Chen, A.~Davis, J.~Dean, M.~Devin,
  S.~Ghemawat, G.~Irving, M.~Isard \emph{et~al.}, ``Tensorflow: A system for
  large-scale machine learning.'' in \emph{OSDI}, vol.~16, 2016, pp. 265--283.

\bibitem{kingma2014adam}
D.~P. Kingma and J.~Ba, ``Adam: A method for stochastic optimization,'' in
  \emph{International Conference on Learning Representations}, ser. ICLR, 2015.

\bibitem{brian_mcfee_2017_848831}
\BIBentryALTinterwordspacing
B.~McFee, C.~Jacoby, E.~J. Humphrey, and W.~Pimenta, ``pescadores/pescador:
  1.1.0,'' Aug. 2017. [Online]. Available:
  \url{https://doi.org/10.5281/zenodo.848831}
\BIBentrySTDinterwordspacing

\bibitem{demvsar2006statistical}
J.~Dem{\v{s}}ar, ``Statistical comparisons of classifiers over multiple data
  sets,'' \emph{Journal of Machine learning research}, vol.~7, no. Jan, pp.
  1--30, 2006.

\bibitem{friedman1937use}
M.~Friedman, ``The use of ranks to avoid the assumption of normality implicit
  in the analysis of variance,'' \emph{Journal of the american statistical
  association}, vol.~32, no. 200, pp. 675--701, 1937.

\bibitem{holm1979simple}
S.~Holm, ``A simple sequentially rejective multiple test procedure,''
  \emph{Scandinavian journal of statistics}, pp. 65--70, 1979.

\end{thebibliography}
%


%

\begin{IEEEbiography}[{\includegraphics[width=1in,height=1.25in,clip,keepaspectratio]{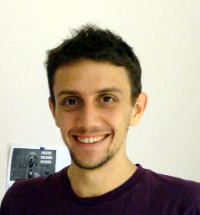}}]{Brian McFee}
is Assistant Professor of Music Technology and Data Science New York University.
He received the B.S.\ degree (2003) in Computer Science from the University of California, Santa Cruz, and M.S.\ (2008) and Ph.D.\ (2012) degrees in Computer Science and Engineering from the University of California, San Diego.
His work lies at the intersection of machine learning and audio analysis.
He is an active open source software developer, and the principal maintainer of the librosa package for audio analysis.
\end{IEEEbiography}

\begin{IEEEbiography}[{\includegraphics[width=1in,height=1.25in,clip,keepaspectratio]{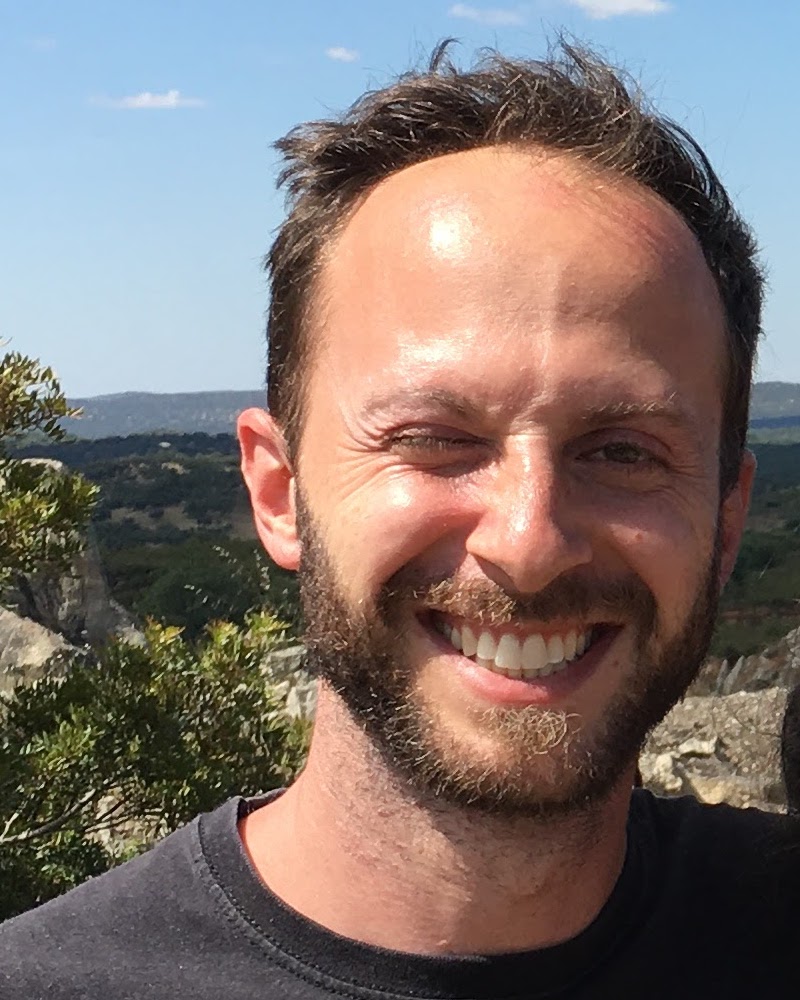}}]{Justin Salamon}
is a Senior Research Scientist at New York University's Music and Audio Research Laboratory and Center for Urban Science and Progress. He received a B.A.\ degree (2007) in Computer Science from the University of Cambridge, UK and M.Sc.\ (2008) and Ph.D.\ (2013) degrees in Computer Science from Universitat Pompeu Fabra, Barcelona, Spain. In 2011 he was a visiting researcher at IRCAM, Paris, France. In 2013 he joined  NYU as a postdoctoral researcher, where he has been a Senior Research Scientist since 2016. His research focuses on the application of signal processing and machine learning to audio signals, with applications in machine listening, music information retrieval, bioacoustics, environmental sound analysis and open source software and data.
\end{IEEEbiography}

\begin{IEEEbiography}[{\includegraphics[width=1in,height=1.25in,clip,keepaspectratio]{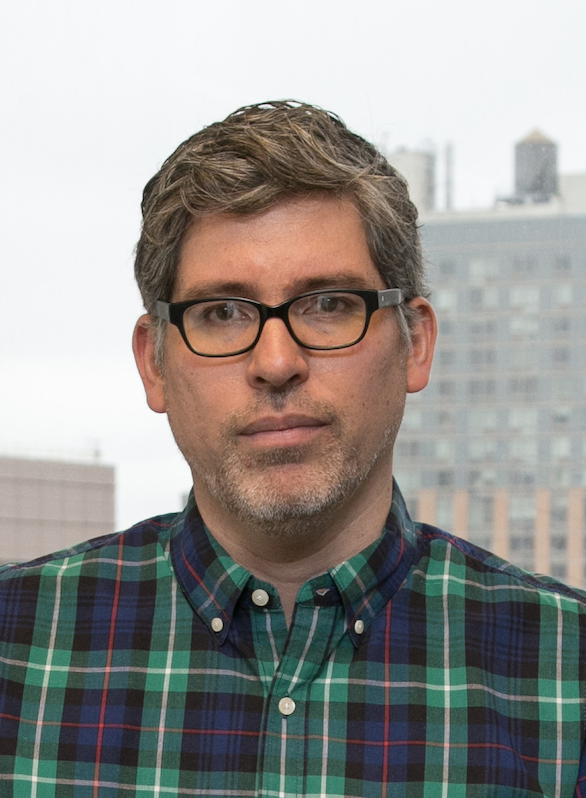}}]{Juan Pablo Bello}
    (SM’16) is Associate Professor of Music Technology and Computer Science \& Engineering at New York University. He received the BEng degree (1998) from Universidad Sim\'{o}n Bol\'{\i}var, Venezuela, and the PhD degree (2003) from Queen Mary, University of London, UK, both in Electronic Engineering. He is director of the Music and Audio Research Lab (MARL), where he leads research in digital signal processing, machine listening and music information retrieval,  topics that he teaches and in which he has published more than 100 papers and articles in books, journals and conference proceedings. His work has been supported by public and private institutions in Venezuela, the UK, and the US, including Frontier and CAREER awards from the National Science Foundation and a Fulbright scholar grant for multidisciplinary studies in France.
\end{IEEEbiography}








\end{document}